%% file: main.tex
\documentclass[nohyperref]{article}
\input{math_commands}

\usepackage{microtype}
\usepackage{graphicx}
\usepackage{subfigure}
\usepackage{float}
\usepackage{soul}

\usepackage{hyperref}

\usepackage[accepted]{icml2023}

\usepackage{amssymb}
\usepackage{mathtools}

\usepackage[capitalize,noabbrev]{cleveref}

\theoremstyle{plain}
\newtheorem{theorem}{Theorem}[section]
\newtheorem{proposition}[theorem]{Proposition}
\newtheorem{lemma}[theorem]{Lemma}

\theoremstyle{definition}

\theoremstyle{remark}

\usepackage[textsize=tiny]{todonotes}

\usepackage{lipsum}

\begin{document}
\twocolumn[
\icmltitle{Adaptive Whitening in Neural Populations with Gain-modulating Interneurons}



\icmlsetsymbol{equal}{*}

\begin{icmlauthorlist}
\icmlauthor{Lyndon R. Duong}{equal,cns}
\icmlauthor{David Lipshutz}{equal,ccn}
\icmlauthor{David J. Heeger}{cns}
\icmlauthor{Dmitri B. Chklovskii}{ccn,ni}
\icmlauthor{Eero P. Simoncelli}{cns,ccn}
\end{icmlauthorlist}

\icmlaffiliation{cns}{Center for Neural Science, New York University;}
\icmlaffiliation{ccn}{Center for Computational Neuroscience, Flatiron Institute;}
\icmlaffiliation{ni}{Neuroscience Institute, NYU School of Medicine}

\icmlcorrespondingauthor{Lyndon R. Duong}{lyndon.duong@nyu.edu}
\icmlcorrespondingauthor{David Lipshutz}{dlipshutz@flatironinstitute.org}

\icmlkeywords{computational neuroscience, recurrent neural networks, adaptation, efficient coding, gain modulation}

\vskip 0.3in
]



\printAffiliationsAndNotice{\icmlEqualContribution} 

\begin{abstract}
Statistical whitening transformations play a fundamental role in many computational systems, and may also play an important role in biological sensory systems.
Existing neural circuit models of adaptive whitening operate by modifying synaptic interactions; however, such modifications would seem both too slow and insufficiently reversible.
Motivated by the extensive neuroscience literature on gain modulation, we propose an alternative model that adaptively whitens its responses by modulating the gains of individual neurons. 
Starting from a novel whitening objective, we derive an online algorithm that whitens its outputs by adjusting the marginal variances of an \textit{overcomplete} set of projections.
We map the algorithm onto a recurrent neural network with fixed synaptic weights and gain-modulating interneurons.
We demonstrate numerically that sign-constraining the gains improves robustness of the network to ill-conditioned inputs, and a generalization of the circuit achieves a form of local whitening in convolutional populations, such as those found throughout the visual or auditory systems.
\end{abstract}

\begin{figure*}[htb]
    \centering
    \includegraphics[width=\textwidth]{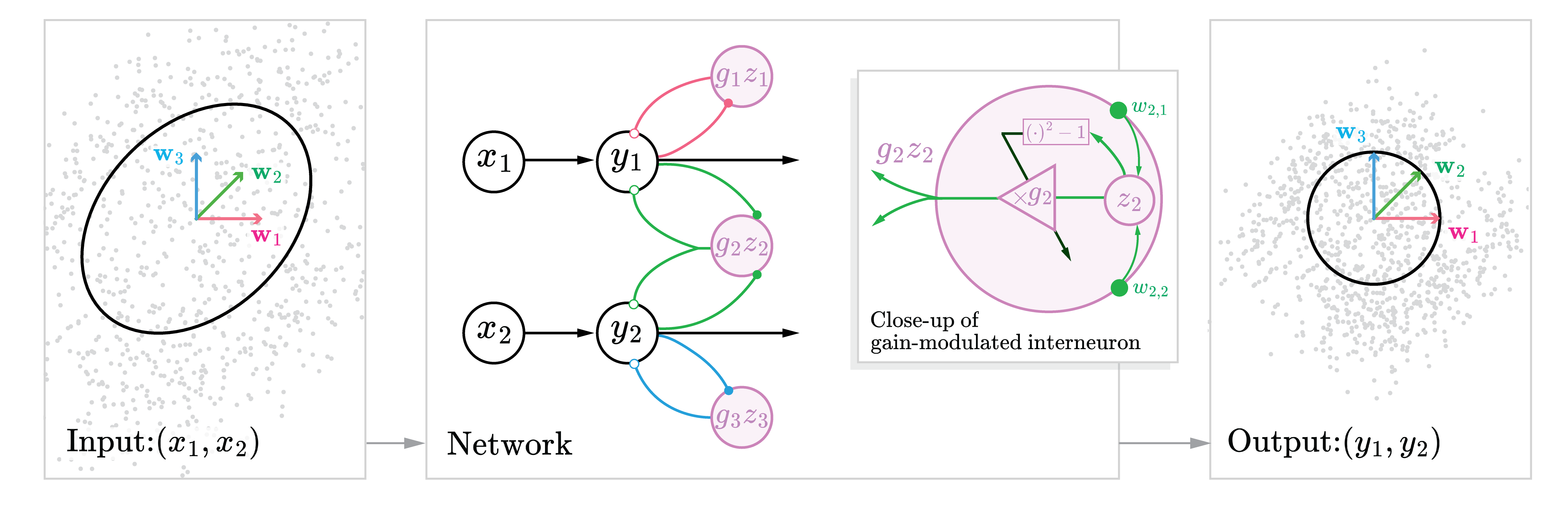}
    \caption{
    Schematic of a recurrent statistical whitening network with 2 primary neurons and 3 interneurons. 
    {\bf Left}: 2D Scatter plot of network inputs $\rvx=[x_1,x_2]^\top$ (e.g. post-synaptic currents), with covariance indicated by the ellipse.
    {\bf Center}:  
    Primary neurons, with outputs $\rvy=[y_1, y_2]^\top$, receive external feedforward inputs, $\rvx$, and recurrent feedback from an overcomplete population of interneurons, $-\sum_{i=1}^3g_iz_i\rvw_i$. 
    Projection vectors $\{{\color{RubineRed}\rvw_1}, {\color{ForestGreen}\rvw_2}, {\color{Cerulean}\rvw_3}\}\in\mathbb{R}^2$ encode feedforward synaptic weights connecting primary neurons to interneuron $i=1,2,3$, with \textit{symmetric} feedback connections. 
    Weight vectors are shown in the left and right panels with corresponding colors.
    In general, the network may require all-to-all connectivity between primary and interneurons; we use a reduced subset of connections here for diagram clarity.
    {\bf Inset:} The $i$\textsuperscript{th} interneuron (e.g. here $i=2$) receives input $z_i=\rvw_i^\top\rvy$, which is multiplied by its gain $g_i$ to produce output $g_iz_i$. 
    Its gain, $g_i$, is adjusted s.t. $\Delta g_i \propto z_i^2-1$. The dark arrow indicates that the gain update operates on a slower time scale.
    {\bf Right:} Scatter plots of the whitened network outputs $\rvy$. 
    Outputs have unit variance along all $\rvw_i$'s, which is equivalent to having identity covariance matrix, i.e., $\rmC_{yy}=\rmI_N$ (black circle).
    }
    \label{fig:schematique}
\end{figure*}

\section{Introduction}

Statistical whitening transformations, in which multi-dimensional inputs are decorrelated and normalized to have unit variance, are common in signal processing and machine learning systems.
For example, they are integral to many statistical factorization methods \citep{olshausen_field1996, bell_sejnowski1997, hyvarinen2000independent}, they provide beneficial preprocessing during neural network training \citep{krizhevsky2009learning}, and they can improve unsupervised feature learning \citep{coates2011analysis}.
More recently, self-supervised learning methods have used decorrelation transformations such as whitening to prevent representational collapse \citep{ermolov2021whitening,zbontar2021barlow,hua2021feature,bardes2021vicreg}.
While whitening has mostly been used for training neural networks in the offline setting, it is also of interest to develop adaptive (run-time) variants that can adjust to dynamically changing input statistics with minimal changes to the network \citep[e.g.][]{mohan_adaptive_2021,hu2021lora}.

Single neurons in early sensory areas of many nervous systems rapidly adjust to changes in input statistics by scaling their input-output gains \citep{adrian1928action}.
This allows neurons to adaptively normalize the variance of their outputs \citep{bonin_statistical_2006, nagel_temporal_2006}, maximizing information transmitted about sensory inputs \citep{barlowpossible1961,laughlinsimple1981,fairhallefficiency2001}.
At the neural \textit{population} level, in addition to variance normalization, adaptive decorrelation and whitening transformations have been observed across species and sensory modalities, including: macaque retina \citep{Atick1992WhatDT}; cat primary visual cortex \citep{muller1999rapid,benucci2013adaptation}; and the olfactory bulbs of zebrafish \citep{friedrich2013neuronal} and mice \citep{giridhar2011timescale,gschwend2015neuronal}. 
These population-level adaptations reduce redundancy in addition to normalizing neuronal outputs, facilitating \textit{dynamic} efficient multi-channel coding \citep{schwartz_natural_2001, barlowadaptation1989}.
However, the mechanisms underlying such adaptive whitening transformations remain unknown, and would seem to require coordinated synaptic adjustments amongst neurons, as opposed to the single neuron case which relies only on gain rescaling. 

Here,  we propose a novel recurrent network architecture for online statistical whitening that exclusively relies on gain modulation.  
Specifically, the primary contributions of our study are as follows:
\begin{enumerate}
    \item We introduce a novel factorization of the (inverse) whitening matrix, using an \textit{overcomplete, arbitrary, but fixed} basis, and a diagonal matrix with statistically optimized entries. This is in contrast with the conventional factorization using the eigendecomposition of the input covariance matrix.
    
    \item  We introduce an unsupervised online learning objective using this factorization to express the whitening objective solely in terms of the \textit{marginal} variances within the overcomplete representation of the input signal. 
    
    \item  We derive a recursive algorithm to optimize the objective, and show that it corresponds to an unsupervised recurrent neural network (RNN), comprised of primary neurons and an auxiliary overcomplete population of interneurons, whose synaptic weights are fixed, but whose gains are adaptively modulated. The network responses converge to the classical symmetric whitening solution without backpropagation.
    
    \item   We show how enforcing non-negativity on the gain modulation provides a novel approach for dealing with ill-conditioned or noisy data.
    Further, we relax the global whitening constraint in our objective and provide a method for \textit{local} decorrelation of convolutional neural populations.
\end{enumerate}

\section{A Novel Objective for Symmetric Whitening}

Consider a neural network with $N$ primary neurons.
For each $t=1,2,\dots$, let $\rvx_t$ and $\rvy_t$ be $N$-dimensional vectors whose components respectively denote the inputs (e.g. post-synaptic currents), and outputs of the primary neurons at time $t$ (\autoref{fig:schematique}). Without loss of generality, we assume the inputs $\rvx_t$ are centered.

\subsection{Conventional objective}

Statistical whitening aims to linearly transform inputs $\rvx_t$ so that the covariance of the outputs $\rvy_t$ is the identity, i.e., 
\begin{align}\label{eq:Cyy}
    \rmC_{yy}=\langle\rvy_t\rvy_t^\top\rangle_t=\rmI_N,
\end{align}
where $\langle\cdot\rangle_t$ denotes the expectation operator over $t$, and $\rmI_N$ denotes the $N\times N$ identity matrix (see \autoref{appendix:notation} for a list of notation used in this work). 

It is well known that whitening is not unique: any orthogonal rotation of a random vector with identity covariance matrix also has identity covariance matrix. 
There are several common methods of resolving this rotational ambiguity, each with their own advantages \citep{kessy2018optimal}. 
Here, we focus on the symmetric whitening transformation, often referred to as Zero-phase Component Analysis (ZCA) whitening or Mahalanobis whitening, which minimizes the mean-squared error between the inputs and the whitened outputs (alternatively, the one whose transformation matrix is symmetric).
The symmetric whitened outputs are the optimal solution to the minimization problem
\begin{align}\label{eq:symmetricobjectivevanilla}
    &\min_{\{\rvy_t\}}\langle\|\rvx_t-\rvy_t\|_2^2\rangle_t\quad\text{s.t.}\quad\langle\rvy_t \rvy_t^\top\rangle_t=\rmI_N,
\end{align}
where $\|\cdot\|_2$ denotes the Euclidean norm on $\R^N$.
Assuming the covariance of the inputs $\rmC_{xx}:=\langle\rvx_t\rvx_t^\top\rangle_t$ is positive definite, the unique solution to the optimization problem in \autoref{eq:symmetricobjectivevanilla} is $\rvy_t=\rmC_{xx}^{-1/2}\rvx_t$ for $t=1,2,\dots$, where $\rmC_{xx}^{-1/2}$ is the symmetric inverse matrix square root of $\rmC_{xx}$ (see Appendix \ref{apdx:optimal}).

Previous approaches to \textit{online} symmetric whitening have optimized \autoref{eq:symmetricobjectivevanilla} by deriving RNNs whose \textit{synaptic weights} adaptively adjust to learn the eigendecomposition of the (inverse) whitening matrix, $\rmC_{xx}^{1/2} = \rmV {\bf \Lambda}^{1/2} \rmV^\top$, where $\rmV$ is an orthogonal matrix of eigenvectors and ${\bf \Lambda}$ is a diagonal matrix of eigenvalues \citep{pehlevan2015normative}.
We propose an entirely different decomposition: $\rmC_{xx}^{1/2} = \rmW \diag{\rvg} \rmW^\top + \rmI_N$, where $\rmW$ is a \textit{fixed} overcomplete matrix of synaptic weights, and $\rvg$ is a vector of \textit{gains} that adaptively adjust to match the whitening matrix. 

\subsection{A novel objective using \textit{marginal} statistics}
We formulate an objective for learning the symmetric whitening transform via gain modulation.
Our innovation exploits the fact that a random vector has identity covariance matrix (i.e., \autoref{eq:Cyy} holds) if and only if it has unit marginal variance along \textit{all possible 1D projections} (a form of tomography; see Related Work). We derive a tighter statement for a finite but \textit{overcomplete} set of at least $K\ge K_N:=N(N+1)/2$ distinct axes (`overcomplete' means that the number of axes exceeds the dimensionality of the input, i.e., $K>N$). 
Intuitively, this equivalence holds because an $N\times N$ symmetric matrix has $K_N$ degrees of freedom, so the marginal variances along $K\ge K_N$ distinct axes are sufficient to constrain an $N\times N$ covariance matrix.
We formalize this equivalence in the following proposition, whose proof is provided in \autoref{appendix:num_projections}.
\begin{proposition} \label{prop:marginal}
Fix $K\geq K_N$. Suppose $\rvw_1,\dots,\rvw_K\in\R^N$ are unit vectors\footnote{The unit-length assumption is imposed, without loss of generality, for notational convenience.} such that 
\begin{equation}\label{eq:spanSN}
    \text{span}(\{\rvw_1\rvw_1^\top,\dots,\rvw_K\rvw_K^\top\})=\sS^N,
\end{equation}
where $\sS^N$ denotes the $K_N$-dimensional vector space of $N\times N$ symmetric matrices. 
Then \autoref{eq:Cyy} holds if and only if the projection of $\rvy_t$ onto each unit vector $\rvw_1,\dots,\rvw_K$ has unit variance, i.e.,
\begin{align}\label{eq:equivalence}
    \langle(\rvw_i^\top\rvy_t)^2\rangle_t=1\quad\text{for}\quad i=1,\dots,K.
\end{align}
\end{proposition}

Assuming \autoref{eq:spanSN} holds, we can interpret the set of vectors $\{\rvw_1,\dots,\rvw_K\}$ as a \textit{frame} \citep[i.e., an overcomplete basis;][]{casazza_introduction_2013} in $\R^N$ such that the covariance of the outputs $\rmC_{yy}$ can be computed from the variances of the $K$-dimensional projection of the outputs onto the set of frame vectors. 
Thus, we can replace the whitening constraint in \autoref{eq:symmetricobjectivevanilla} with the equivalent \textit{marginal variance} constraint to obtain the following objective:
\begin{align}\label{eq:symmetricobjective}
    \min_{\{\rvy_t\}}\langle\|\rvx_t-\rvy_t\|_2^2\rangle_t \quad \text{s.t.}&\quad
    \text{\autoref{eq:equivalence} holds}.
\end{align}

\section{An RNN with Gain Modulation for Adaptive Symmetric Whitening}
In this section, we derive an online algorithm for solving the optimization problem in \autoref{eq:symmetricobjective} and map the algorithm onto an RNN with adaptive gain modulation. 
Assume we have an overcomplete frame $\{\rvw_1,\dots,\rvw_K\}$ in $\R^N$ satisfying \autoref{eq:spanSN}.
We concatenate the frame vectors into an $N\times K$ synaptic weight matrix $\rmW:=[\rvw_1,\dots,\rvw_K]$. 
In our network, primary neurons project onto a layer of $K$ interneurons via the synaptic weight matrix to produce the $K$-dimensional vector $\rvz_t:=\rmW^\top\rvy_t$, encoding the interneurons' post-synaptic inputs at time $t$ (\autoref{fig:schematique}). 
We emphasize that the synaptic weight matrix $\rmW$ remains \textit{fixed}.

\subsection{Enforcing the marginal variance constraints with scalar gains}

We introduce Lagrange multipliers $g_1,\dots,g_K\in\R$ to enforce the $K$ constraints in \autoref{eq:equivalence}.
These are concatenated as the entries of a $K$-dimensional vector $\rvg:=[g_1,\dots,g_K]^\top\in\R^K$, and express the whitening objective as a saddle point optimization:
\begin{align}\label{eq:Lagrange}
    &\max_\rvg\min_{\{\rvy_t\}}\langle\ell(\rvx_t,\rvy_t,\rvg)\rangle_t,\\
    &\text{where}\enspace\ell(\rvx,\rvy,\rvg):=\|\rvx-\rvy\|_2^2+\sum_{i=1}^Kg_i\left\{(\rvw_i^\top\rvy)^2-1\right\}.\nonumber
\end{align}
Here, we have exchanged the order of maximization over $\rvg$ and minimization over $\rvy_t$, which is justified because $\ell(\rvx_t,\rvy_t,\rvg)$ satisfies the saddle point property with respect to $\rvy$ and $\rvg$, see \autoref{apdx:saddle}. 

In our RNN implementation, there are $K$ interneurons and $g_i$ corresponds to the multiplicative gain associated with the $i$\textsuperscript{th} interneuron, so that its output at time $t$ is $g_iz_{i,t}$ (\autoref{fig:schematique}, Inset).
\autoref{eq:Lagrange}, shows that the gain of the $i$\textsuperscript{th} interneuron, $g_i$, encourages the marginal variance of $\rvy_t$ along the axis spanned by $\rvw_i$ to be unity.
Importantly, the gains are not hyper-parameters, but rather they are optimization variables which statistically whiten the outputs $\{\rvy_t\}$, preventing the neural outputs from trivially matching the inputs $\{\rvx_t\}$.

\subsection{Deriving RNN neural dynamics and gain updates}
\label{ssec:offline_algorithm}

To solve \autoref{eq:Lagrange} in the online setting, we assume there is a time-scale separation between `fast' neural dynamics and `slow' gain updates, so that at each time step the neural dynamics equilibrate before the gains are adjusted. 
This allows us to perform the inner minimization over $\{\rvy_t\}$ before the outer maximization over the gains $\rvg$.
This is consistent with biological networks in which a given neuron's responses operate on a much faster time-scale than its intrinsic input-output gain, which is driven by slower processes such as changes in Ca\textsuperscript{2+} concentration gradients and Na\textsuperscript{+}-activated K\textsuperscript{+} channels \citep{wang2003adaptation, ferguson2020mechanisms}.

\subsubsection{Fast neural activity dynamics}
For each time step $t=1,2,\dots$, we minimize the objective $\ell(\rvx_t,\rvy_t,\rvg)$ over $\rvy_t$ by recursively running gradient-descent steps to equilibrium:
\begin{align}\label{eq:dydtau}
    \rvy_t&\gets\rvy_t-\frac\gamma 2\nabla_{\rvy}\ell(\rvx_t,\rvy_t(\tau),\rvg) \nonumber\\ 
    &=\rvy_t+\gamma\left\{\rvx_t-\rmW(\rvg\circ\rvz_t)-\rvy_t\right\}, 
\end{align}
where $\gamma>0$ is a small constant, $\rvz_t=\rmW^\top\rvy_t$, the circle `$\circ$' denotes the Hadamard (element-wise) product, $\rvg \circ \rvz_t$ is a vector of $K$ gain-modulated interneuron outputs, and we assume the primary cell outputs are initialized at zero. 

We see from the right-hand-side of \autoref{eq:dydtau} that the `fast' dynamics of the primary neurons are driven by three terms (within the curly braces): 1) constant feedforward external input $\rvx_t$; 2) recurrent gain-modulated feedback from interneurons $-\rmW(\rvg\circ\rvz_t)$; and 3) a leak term $-\rvy_t$. 
Because the neural activity dynamics are linear, we can analytically solve for their equilibrium (i.e. steady-state), $\bar\rvy_t$, by setting the update in \autoref{eq:dydtau} to zero:
\begin{align}\label{eq:y_steadystate}
    \bar\rvy_t&=\left[\rmI_N+\rmW\diag{\rvg}\rmW^\top\right]^{-1}\rvx_t \nonumber \\
    &= \left[\rmI_N+\sum_{i=1}^Kg_i\rvw_i\rvw_i^\top\right]^{-1}\rvx_t,
\end{align}
where $\diag{\rvg}$ denotes the $K\times K$ diagonal matrix whose $(i,i)$\textsuperscript{th} entry is $g_i$, for $i=1,\dots,K$. 
The equilibrium feedforward interneuron inputs are then given by 
\begin{align}\label{eq:z_steadystate}
    \bar\rvz_t=\rmW^\top\bar\rvy_t.
\end{align}
The gain-modulated outputs of the $K$ interneurons, $\rvg \circ \rvz_t$, are then projected back onto the primary cells via symmetric weights, $-\rmW$~(\autoref{fig:schematique}). 
After $\rvg$ adapts to optimize \autoref{eq:Lagrange} (provided Proposition~\ref{prop:marginal} holds), the matrix within the brackets in \autoref{eq:y_steadystate} will equal $\rmC_{xx}^{1/2}$, and the circuit's equilibrium responses are symmetrically whitened.
The result is a novel \textit{overcomplete} symmetric matrix factorization in which $\rmW$ is arbitrary and fixed, while $\rmC_{xx}^{1/2}$ is adaptively learned and encoded in the gains $\rvg$. 

\subsubsection{Slow gain dynamics}
After the fast neural activities reach steady-state, the interneuron gains are updated with a stochastic gradient-ascent step with respect to $\rvg$:
\begin{align}\label{eq:gupdate}
    \rvg&\gets\rvg+\frac\eta2\nabla_\rvg\ell(\rvx_t,\bar\rvy_t,\rvg) \nonumber\\
    &=\rvg+\eta\left(\bar\rvz_t^{\circ2}-{\bf 1}\right),
\end{align}
where $\eta>0$ is the learning rate, $\bar\rvz_t^{\circ2}=[\bar z_{t,1}^2,\dots,\bar z_{t,K}^2]^\top$,  and ${\bf 1}=[1,\dots,1]^\top$ is the $K$-dimensional vector of ones\footnote{\autoref{appendix:generalizations} generalizes the gain update to allowing for temporal-weighted averaging of the variance over past samples.}.
Remarkably, the update to the $i$\textsuperscript{th} interneuron's gain $g_i$ (\autoref{eq:gupdate}) depends only on the online estimate of the \textit{variance} of its equilibrium input $\bar z_{t,i}^2$, and its distance from 1 (i.e. the target variance). 
Since the interneurons adapt using local signals, this circuit is a suitable candidate for hardware implementations using low-power neuromorphic chips \citep{pehlevan_neuroscience-inspired_2019}.
Intuitively, each interneuron adjusts its gain to modulate the amount of suppressive (inhibitory) feedback onto the joint primary neuron responses. 
In \autoref{appendix:psd_frame}, we provide conditions under which $\rvg$ can be solved analytically.
Thus, while statistical whitening inherently involves a transformation on a joint density, our solution operates solely using single neuron gain changes in response to \textit{marginal} statistics of the joint density.

\subsubsection{Online unsupervised algorithm}
By combining Equations~\ref{eq:dydtau} and \ref{eq:gupdate}, we arrive at our online RNN algorithm for adaptive whitening via gain modulation (Algorithm \ref{alg:online}). 
We also provide batched and offline versions of the algorithm in \autoref{appendix:algorithms}.

\begin{algorithm}[htb]
\caption{Adaptive whitening via gain modulation}
\label{alg:online}
\begin{algorithmic}[1]
\STATE {\bfseries Input:} Centered inputs $\rvx_1,\rvx_2,\dots\in\R^N$
\STATE {\bfseries Initialize:} $\rmW\in\R^{N\times K}$; $\rvg\in\R^K$; $\eta,\gamma>0$ \\
 \FOR{$t=1,2,\dots$} 
\STATE $\rvy_t\gets{\bf 0}$
    \WHILE{not converged}  
        \STATE $\rvz_t\gets\rmW^\top\rvy_t$ 
        \STATE $\rvy_t\gets\rvy_t+\gamma\left\{\rvx_t-\rmW(\rvg\circ\rvz_t)-\rvy_t\right\}$ 
    \ENDWHILE
    \STATE $\rvg\gets\rvg+\eta\left({\rvz}_t^{\circ2}- {\bf 1} \right)$ 
 \ENDFOR
\end{algorithmic}
\end{algorithm}

There are two points worth noting about this network:
    1) $\rmW$ remains \textit{fixed} in Algorithm \ref{alg:online}. 
    Instead, $\rvg$ adapts to statistically whiten the outputs.  
  2) In practice, since network dynamics are linear, we can bypass the inner loop (the fast dynamics of the primary cells, lines 5--8), by directly computing $\bar{\rvy}_t$,  and $\bar{\rvz}_t$ (Eqs.~\ref{eq:y_steadystate},~\ref{eq:z_steadystate}).
 

\section{Numerical Experiments and Applications}

We provide different applications of our adaptive symmetric whitening network via gain modulation, emphasizing that gain adaptation is distinct from, and \textit{complementary to}, synaptic weight learning (i.e.\ learning $\rmW$).
We therefore side-step the goal of learning the frame $\rmW$, and assume it is fixed (for example, through longer time scale learning).
This allows us to decouple and analyze the general properties of our proposed gain modulation framework, independently of the choice of frame.
Python code for this study can be located at {\footnotesize \hyperlink{https://github.com/lyndond/frame_whitening}{\texttt{github.com/lyndond/frame\_whitening}}}.

We evaluate the performance of our adaptive whitening algorithm using the matrix operator norm, $\Vert \cdot \Vert_{\text{Op}}$, which measures the largest eigenvalue,
\begin{align*}
   \text{Error} := \Vert\rmC_{yy} - \rmI_N \Vert_{\text{Op}}.
\end{align*}
As a performance criterion, we use $\Vert \rmC_{yy} - \rmI_N \Vert_{\text{Op}} \leq 0.1$, the point at which the principal axes of $\rmC_{yy}$ are within 0.1 of unity. 
Geometrically, this means the  ellipsoid corresponding to the covariance matrix lies between the circles with radii 0.9 and 1.1.

For visualization of output covariance matrices, we plot 2D ellipses representing the 1-standard deviation probability level-set contour of the density.
These ellipses are defined by the set of points $\{\|\rmC_{yy}^{1/2}\rvv\|\rvv:\|\rvv\|=1\}$.

\subsection{Adaptive symmetric whitening via gain modulation}
\label{ssec:algorithm_validation}

\begin{figure}[tb]
\begin{center}
\centerline{\includegraphics[width=.8\columnwidth]{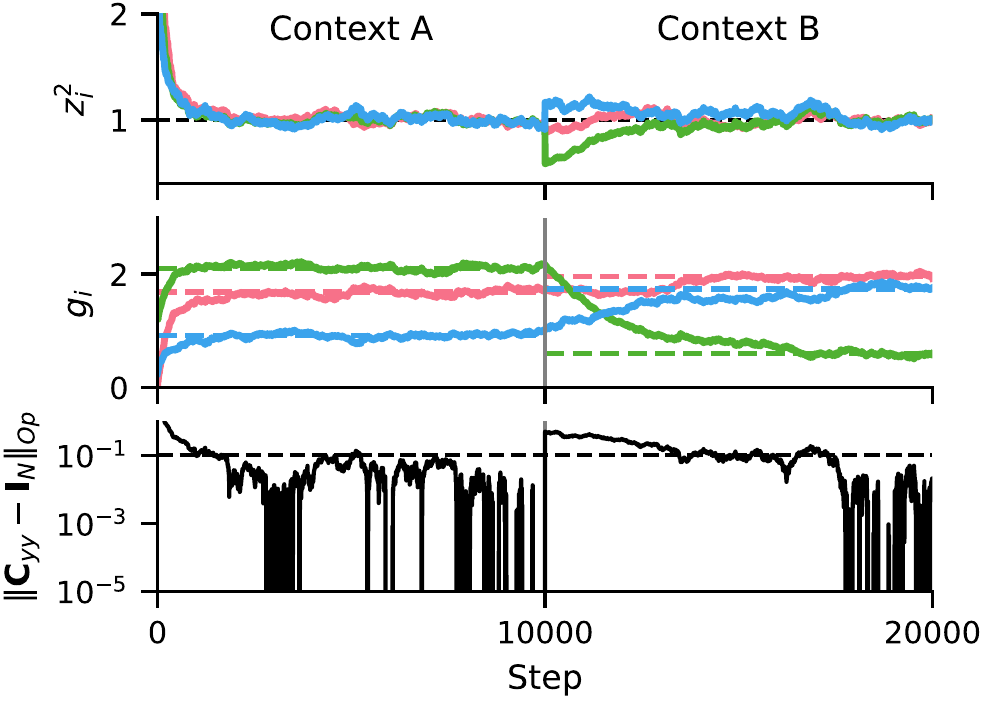}}
\caption{
    Network from \autoref{fig:schematique} (with corresponding colors; $N{=}2$, $K{=}K_N{=}3$, $\eta{=}$2E-3) adaptively whitening samples from two randomly generated statistical contexts online (10K steps each). {\bf Top:} Marginal variances measured by interneurons approach 1 over time.
    {\bf Middle:} Dynamics of interneuron gains, which are applied to $z_i$ before feeding back onto the primary cells. 
    Dashed lines are optimal gains (\autoref{appendix:psd_frame}).
    {\bf Bottom:} Error over time, as measured by the maximal difference between the standard deviation along the principal axes of $\rmC_{yy}$ and unity.
}
\label{fig:loss_go_down}
\end{center}
\vskip -0.1in
\end{figure}

We first demonstrate that our algorithm successfully whitens its outputs.
We initialize a network with fixed interneuron weights, $\rmW$, corresponding to the frame illustrated in \autoref{fig:schematique} ($N{=}2$, $K{=}K_N{=}3$).
\autoref{fig:loss_go_down} shows the network adapting to inputs from two successively-presented contexts with randomly-generated underlying input covariances $\rmC_{xx}$ (10K gain update steps each).
As update steps progress, all marginal variances converge to unity, as expected from the objective (top panel).
Since the number of interneurons satisfies $K{=}K_N$, the optimal gains to achieve symmetric whitening can be solved analytically (\autoref{appendix:psd_frame}), and are shown in the middle panel (dashed lines).

\autoref{fig:loss_go_down} illustrates the \textit{online, adaptive} nature of the network; it whitens inputs from novel statistical contexts at run-time, without supervision.
By Proposition~\ref{prop:marginal}, measuring unit variance along $K_N$ unique axes, as in this example, guarantees that the underlying joint density is statistically white.
Indeed, the whitening error (bottom panel), approaches zero as all $K_N$ marginal variances approach 1.
Thus, with interneurons monitoring their respective \textit{marginal} input variances $z_i^2$, and re-scaling their gains to modulate feedback onto the primary neurons, the network adaptively whitens its outputs in each context.

\subsection{Algorithmic convergence rate depends on $\rmW$}\label{ssec:convergence}
Our model assumes that the frame, $\rmW$, is fixed and known (e.g., optimized via pre-training or development).
This distinguishes our method from existing symmetric whitening methods, which typically operate by estimating and transforming to the eigenvector basis.
By contrast, our network obviates learning the principal axes of the data altogether, and instead uses a statistical sampling approach along the fixed set of measurement axes spanned by $\rmW$.
While the result expressed in Proposition \ref{prop:marginal} is exact, and the \textit{optimal solution} to the whitening objective \autoref{eq:symmetricobjective} is independent of ${\rmW}$ (provided \autoref{eq:spanSN} holds), we hypothesize that the \textit{algorithmic convergence rate} would depend on $\rmW$.

\autoref{fig:grassmann} summarizes an experiment assessing the convergence rate of different networks whitening inputs with a random covariance, $\rmC_{xx}$, with $N=2$ (the results are consistent when $N>2$).
We initialize three kinds of frames $\rmW \in \R^{N \times K_N}$ with 100 repetitions each: `{\bf Random}', a frame with i.i.d.\ Gaussian entries; `{\bf Optimized}', a randomly initialized frame whose columns are then optimized to have minimum mutual coherence and cover the ambient space; and `{\bf Spectral}', a frame whose first $N$ columns are the eigenvectors of the data and the remaining $K_N-N$ columns are zeros.
For clarity, we remove the effects of input sampling stochasticity by running the offline version of our network, which assumes having direct access to the input covariance (\autoref{appendix:algorithms}); the online version is qualitatively similar.

\begin{figure}[htb]
\begin{center}
\centerline{\includegraphics[width=\columnwidth]{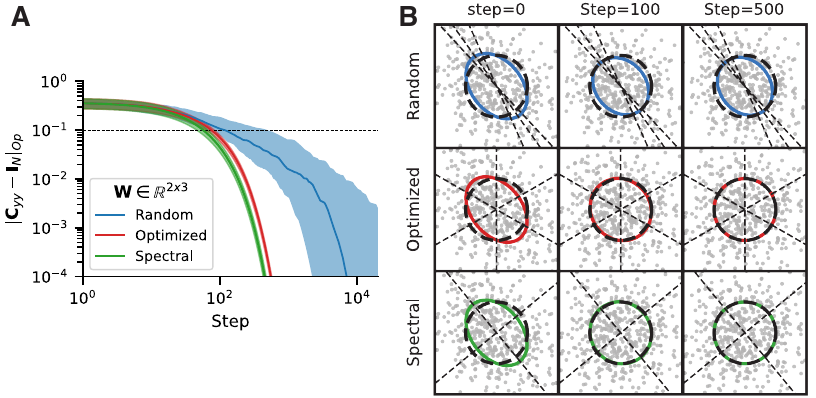}}
\caption{
Convergence rate depends on structure of $\rmW$. For each network, $\eta{=}$1E-2.
{\bf A:} Error over time.
Curves are median and [25\%, 75\%] quantile regions over 100 repeats.
Dashed line indicates when the principal axes of 1-standard deviation ellipse representing $\rmC_{yy}$ are within 0.1 of unity.
{\bf B:} Scatter plots and covariance ellipses of $\rvy$ for a single experiment with each frame type at different steps. Gray dashed lines are axes spanned by $\rmW$.
}
\label{fig:grassmann}
\end{center}
\vskip -0.3in
\end{figure}

When the input distribution is known, then using the input covariance eigenvectors, as with the Spectral frame, defines a bound on achievable performance, converging faster, on average, than the Random and Optimized frames (\autoref{fig:grassmann}A,B).
This is because the frame is aligned with the input covariance's principal axes, and a simple gain scaling along those directions is sufficient to achieve a whitened response.
We find that the networks with Optimized frames converge at similar rates to those with Spectral frames, despite the frame vectors not being aligned with the principal axes of the data (\autoref{fig:grassmann}B).
Comparing the Random to Optimized frames gives a better understanding of how one might choose a frame in the more realistic scenario when the input distribution is unknown.
The networks with Optimized frames systematically converge faster than Random frames.
Thus, when the input distribution is unknown, we empirically find that the convergence rate of Algorithm~\ref{alg:online} benefits from a frame that is optimized to splay the ambient space. 
Increased coverage of the space by the frame vectors facilitates whitening with our gain re-scaling mechanism.
Sec.~\ref{ssec:conv} elaborates on how underlying signal structure can be exploited to inform more efficient choices of frames.

\subsection{Implicit sparse gating via gain modulation}

\begin{figure}[htb]
\begin{center}
\centerline{\includegraphics[width=\columnwidth]{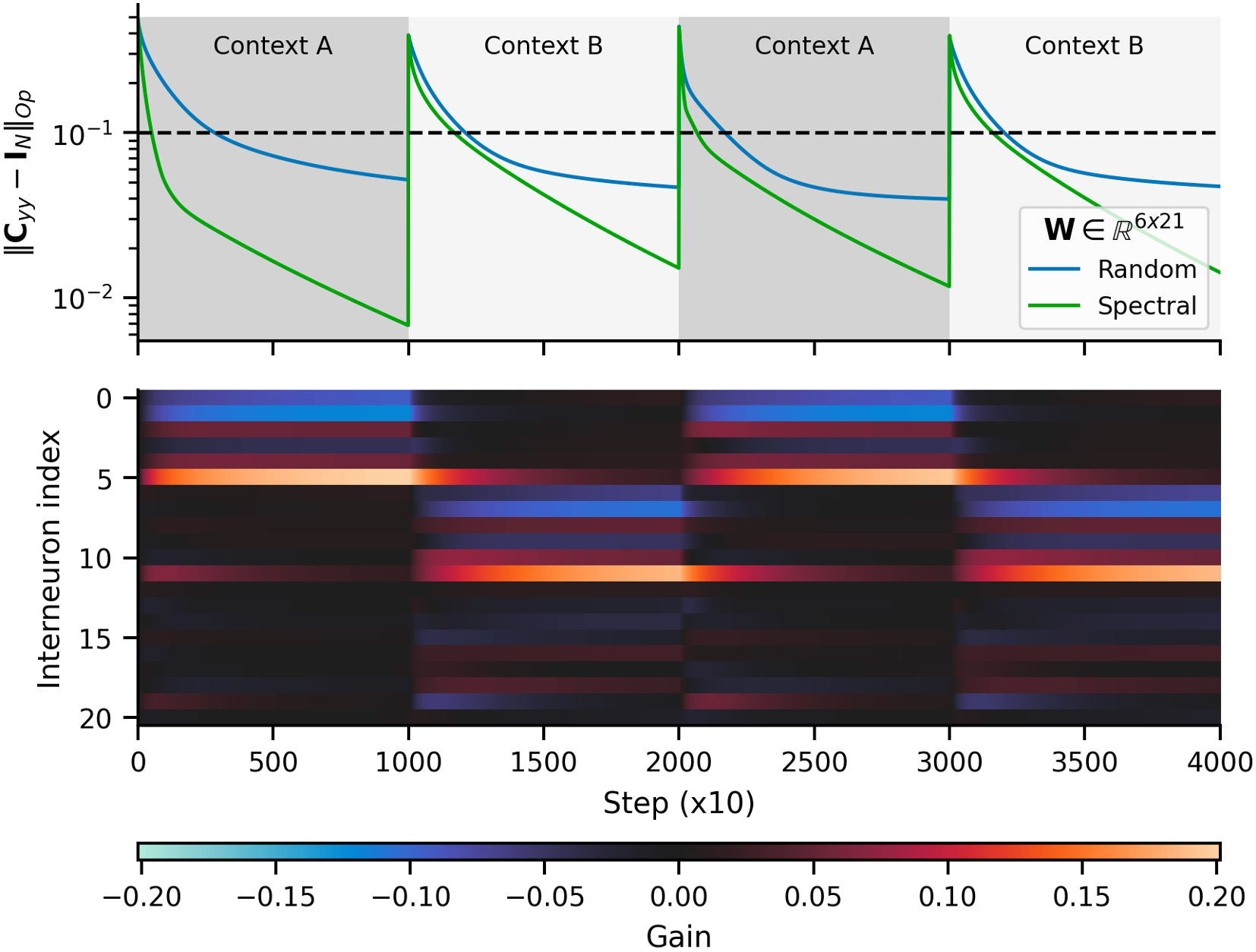}}
\caption{
Gain modulation as a fast implicit sparse gating mechanism.
{\bf Top}: Error over time for Spectral vs. Random networks ($N{=}6$; $K{=}K_N{=}21$; $\eta{=}$1E-3) adapting to 2 alternating statistical contexts with different input covariances.
Dashed line indicates when the principal axes of 1-standard deviation ellipsoid representing $\rmC_{yy}$ are within 0.1 of unity.
{\bf Bottom}: Gains act as implicit context switches, sparsely gating the respective eigenbases embedded in the Spectral frame to optimally whiten each context.
}
\label{fig:contexts}
\end{center}
\vskip -0.3in
\end{figure}

Motivated by the findings in Sec~\ref{ssec:convergence}, and concepts from sparse coding \citep{olshausen_field1996}, we explore how adaptive gain modulation can complement or augment a `pre-trained' network with context-dependent weights.
\autoref{fig:contexts} shows an experiment using either a pre-trained Spectral, or Random $\rmW$ ($N{=}6$, $K{=}K_N{=}21$) adaptively whitening inputs from two random, alternating statistical contexts, A and B, for 10K steps each.
The first and second $N$ columns of the Spectral frame are the eigenvectors of context A and B's covariance matrix, respectively, and the remaining elements are random i.i.d.\ Gaussian; the Random frame has all i.i.d. Gaussian elements.
\autoref{fig:contexts} (top panel) shows that both networks successfully adapt to whiten the inputs from each context, with the Spectral frame converging faster than the Random frame (as in Sec~\ref{ssec:convergence}).

Inspecting the Spectral frame's $K$ interneuron gains during run-time (bottom panel) reveals that they sparsely `select' the frame vectors corresponding to the eigenvectors of each respective condition (indicated by the blue/red intensity).
This effect arises \textit{without} a sparsity penalty or modifying the objective.
Gain modulation thus \textit{sparsely gates} context-dependent information without an explicit context signal.

\subsection{Normalizing ill-conditioned data}

\begin{figure}[htb]
\begin{center}
\centerline{\includegraphics[width=.9\columnwidth]{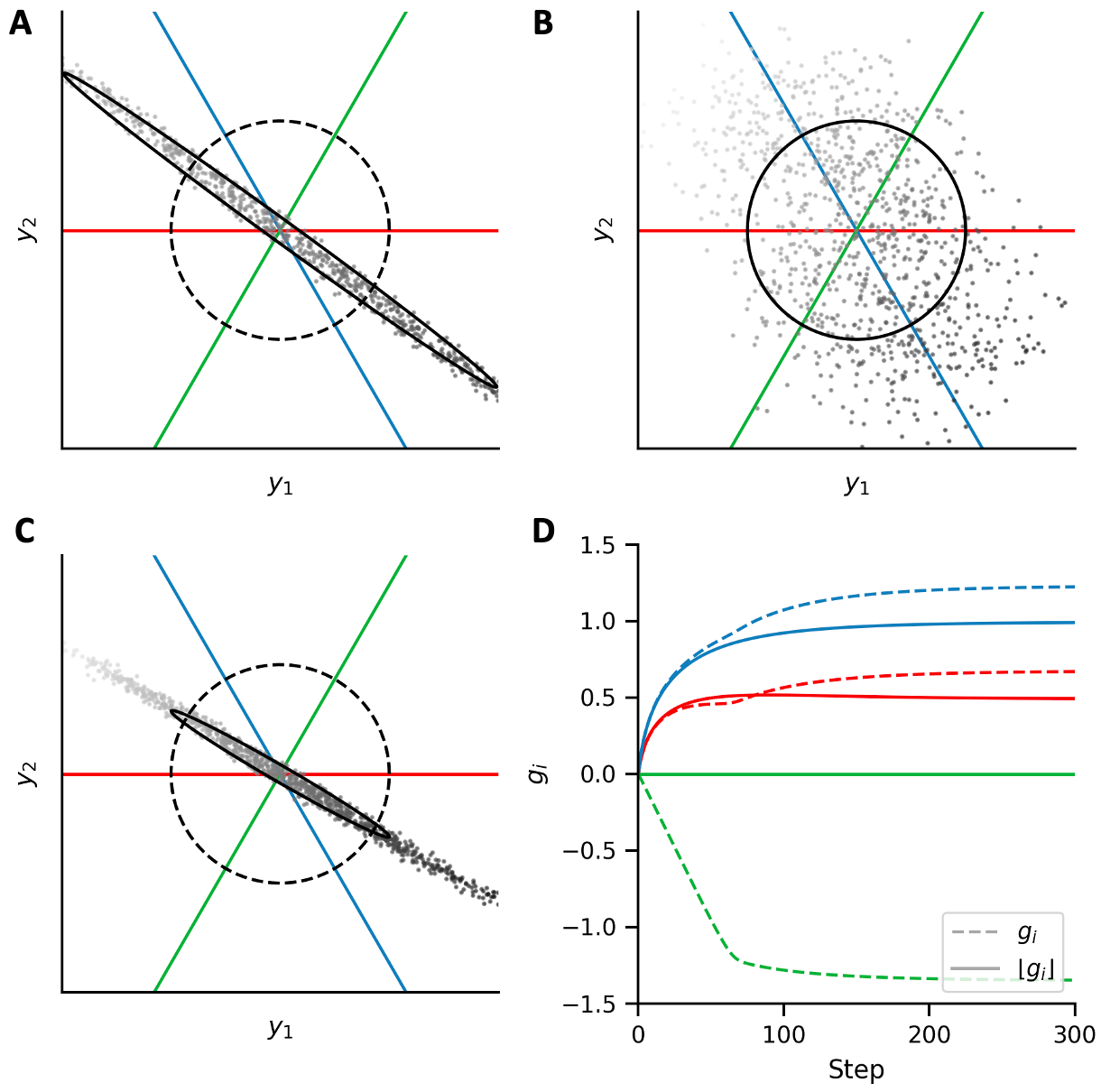}}
\caption{
Two networks ($N{=}2$, $K{=}3$, $\eta{=}0.02$) whitening ill-conditioned inputs.
{\bf A:} Outputs without whitening. 2D scatterplot of a non-Gaussian density whose underlying signal lies close to a latent 1D axis. 
Many points lie outside of the axis limits in this panel.
Signal magnitude along that axis is denoted by the grayscale gradient.
The 1-standard deviation covariance matrix is depicted as a black ellipse.
Colored lines are axes spanned by Optimal frame (see Sec~\ref{ssec:convergence}).
{\bf B:} Symmetric whitening boosts noise along the uninformative direction.
{\bf C:} Modulating gains according to Eq.~\ref{eq:grectified} rescales the data \textit{without} amplifying noise.
{\bf D:} Gains updated with Eq.~\ref{eq:gupdate}  vs. Eq.~\ref{eq:grectified}. Colors correspond to frame axes in panels A--C.
}
\label{fig:low_rank}
\end{center}
\vskip -0.3in
\end{figure}
Foundational work by \citet{Atick1992WhatDT} showed that neural populations in the retina may encode visual inputs by optimizing mutual information in the presence of noise.
For natural images with $1/f$ spectra, the optimal transform is approximately a product of a whitening filter and a low-pass filter. 
This is a particularly effective solution because when inputs are low-rank, $\rmC_{xx}$ is ill-conditioned (\autoref{fig:low_rank}A), and classical whitening leads to noise amplification along axes with small variance. 
In this section, we show how a simple modification to the objective allows our gain-modulating network to handle these types of inputs.

We prevent amplification of inputs below a certain variance threshold by replacing the unit marginal variance equality constraints with upper bound constraints\footnote{
We set the threshold to 1 to remain consistent with the whitening objective, but it can be any arbitrary variance.}:
\begin{align}\label{eq:ineq_constraint}
    \langle(\rvw_i^\top\rvy_t)^2\rangle_t\le1\quad\text{for}\quad i=1,\dots,K.
\end{align}
Our modified network objective then becomes
\begin{align}\label{eq:symmetricupper}
    \min_{\{\rvy_t\}}\langle\|\rvx_t-\rvy_t\|_2^2\rangle_t\quad\text{s.t.}\quad\text{\autoref{eq:ineq_constraint} holds.}
\end{align}
Intuitively, if the projected variance along a given direction is already less than or equal to unity, then it will not affect the overall loss.
Interneuron gain should accordingly \textit{stop adjusting} once the marginal variance along its projection axis is less than or equal to one.
To enforce these upper bound constraints, we introduce gains as Lagrange multipliers, but restrict the domain of $\rvg$ to be the non-negative orthant $\R_+^K$, resulting in non-negative optimal gains:
\begin{align}\label{eq:Lagrange_nn}
    \max_{\rvg \in \R^K_+}\min_{\{\rvy_t\}}\langle\ell(\rvx_t,\rvy_t,\rvg)\rangle_t,
\end{align}
where $\ell(\rvx,\rvy,\rvg)$ is defined as in \autoref{eq:Lagrange}.
At each time step $t$, we optimize \autoref{eq:Lagrange_nn} by first taking gradient-descent steps with respect to $\rvy_t$, resulting in the same neural dynamics (\autoref{eq:dydtau}) and equilibrium solution (\autoref{eq:y_steadystate}) as before. 
To update $\rvg$, we modify \autoref{eq:gupdate} to take a \textit{projected} gradient-ascent step with respect to $\rvg$:
\begin{align}\label{eq:grectified}
    \rvg \leftarrow \lfloor \rvg + \eta (\bar{\rvz}_t^{\circ 2} - {\bf 1}) \rfloor
\end{align}
where $\lfloor \cdot \rfloor$ denotes the element-wise half-wave rectification operation that projects its inputs onto the non-negative orthant $\R_+^K$, i.e., $\lfloor \rvv\rfloor := [\max(v_1,0),\dots,\max(v_K,0)]^\top$.

\autoref{fig:low_rank} shows a simulation of a network whitening ill-conditioned inputs with an Optimized frame ($N{=}2$, $K{=}K_N$; see Sec.~\ref{ssec:convergence}) where gains are either unconstrained (\autoref{eq:gupdate}), or rectified (\autoref{eq:grectified}). 
We observe that these two models converge to two different solutions (\autoref{fig:low_rank}B, C).
When $g_i$ is unconstrained, the network achieves global whitening, as before, but in doing so it amplifies noise along the axis orthogonal to the latent signal axis.
The gains constrained to be non-negative converged to different values than the unconstrained gains (\autoref{fig:low_rank}D), with one of them (green) converging to zero rather than becoming negative.
In general, with constrained $g_i$, the whitening error network converges to a non-zero value (see \autoref{appendix:nonneg} for details).
Thus, with a non-negative constraint, the network normalizes the responses $\rvy$,  and \textit{does not amplify the noise}.
In \autoref{appendix:nonneg} we show additional cases that provide further geometric intuition on differences between symmetric whitening with and without non-negative constrained gains.

\subsection{Gain modulation enables local spatial decorrelation}\label{ssec:conv}

\begin{figure*}[htb]
\begin{center}
\centerline{\includegraphics[width=\textwidth]{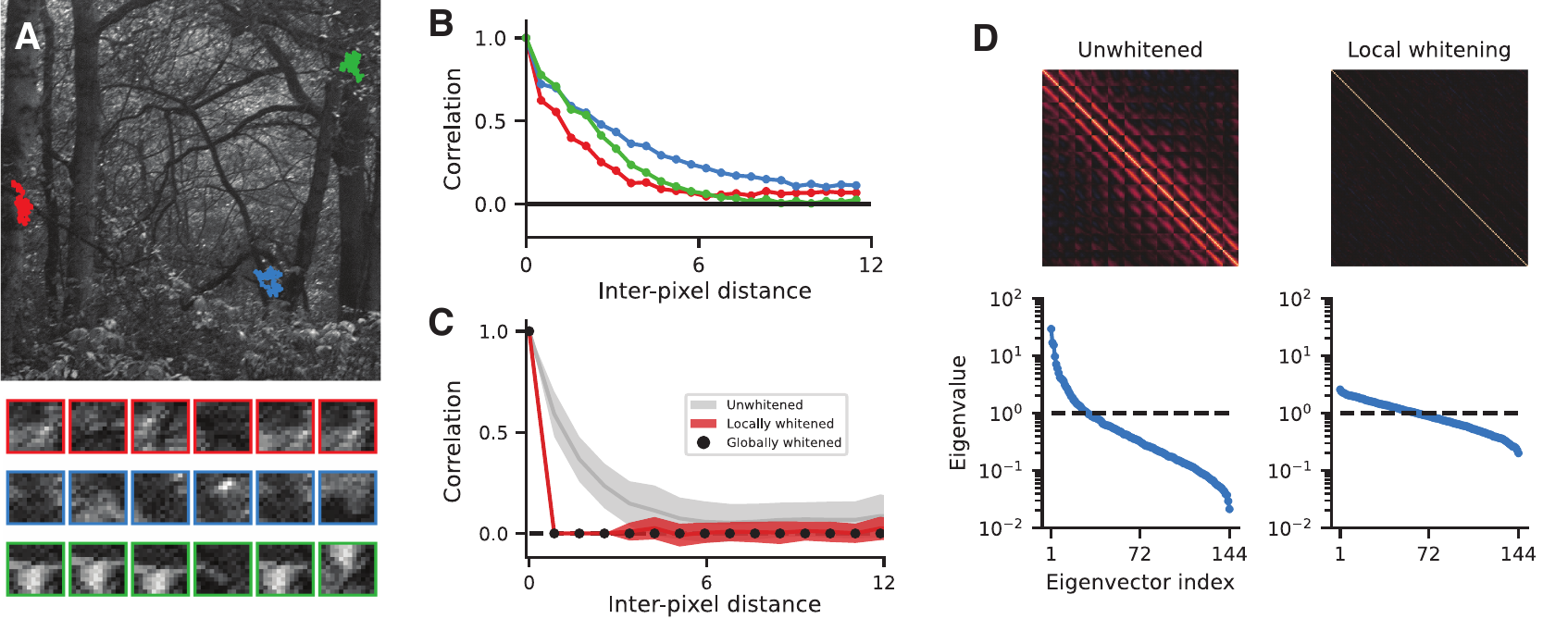}}
    \caption{ 
    Local spatial whitening.
{\bf A)} Large grayscale image from which $12{\times}12$ image patch samples are drawn.  Colors represent random-walk sampling from regions of the image corresponding to contexts with different underlying statistics.
Six samples from each context are shown below.
{\bf B)} Without whitening, pixel correlations decay rapidly with spatial distance in each context, suggesting that local whitening may be effective.
{\bf C)} Binned pairwise output pixel correlation of patches from the red context before (gray) and after global (black dots) vs.\ local whitening with overlapping $4{\times}4$ neighborhoods (red). 
Shaded regions represent standard deviations.
{\bf D)} Top: Correlation matrices of flattened patches from the red context before whitening (left), and after local symmetric whitening (right). 
Both panels use the same color scale.
Bottom: Corresponding covariance eigenspectra. Dashed lines are spectra after global whitening.
}
\label{fig:conv2d}
\end{center}
\vskip -0.4in
\end{figure*}

Requiring $K_N$ interneurons to guarantee a statistically white output (Proposition~\ref{prop:marginal}) becomes prohibitively costly for high-dimensional inputs: the number of interneurons scales as $\mathcal{O}(N^2)$.
This leads us to ask: how many interneurons are needed in practice?
For natural sensory inputs such as images, it is well-known that inter-pixel correlation is highly structured, decaying as a function of distance.
We simulate an experiment of visual gaze fixations and micro-saccadic eye movements using a Gaussian random walk, drawing $12{\times}12$ patch samples from a region of a natural image \citep[\autoref{fig:conv2d}A; ][]{hateren_schaaf_1998}; this can be interpreted as a form of video-streaming dataset where each frame is a patch sample.
We repeat this for different randomly selected regions of the image (\autoref{fig:conv2d}A colors).
The image content of each region is quite different, but the inter-pixel correlation within each context consistently falls rapidly with distance (\autoref{fig:conv2d}B).

We \textit{relax} the $\mathcal{O}(N^2)$ marginal variance constraint to instead whiten \textit{spatially local neighborhoods} of primary neurons whose inputs are the image patches.
We construct a frame $\rmW$ that exploits spatial structure in the image patches, and spans $K<K_N$ axes in $\R^N$.
$\rmW$ is convolutional, such that \textit{overlapping} neighborhoods of $4 \times 4$ primary neurons are decorrelated, each by a population of interneurons that is `overcomplete' with respect to that neighborhood (see \autoref{appendix:conv} for details).
Importantly, taking into account local structure dramatically reduces the interneuron complexity from $\mathcal{O}(N^2) \rightarrow \mathcal{O}(N)$, thereby making our framework practically feasible for high-resolution image inputs and video streams.
This frame is still overcomplete ($K>N$), but because $K<K_N$, we no longer guarantee at equilibrium that $\rmC_{yy}=\rmI_N$ (Proposition~\ref{prop:marginal}).

After the network converges to the inputs drawn from the red context (\autoref{fig:conv2d}C): i) inter-pixel correlations drop within the region specified by the local neighborhood; and ii) surprisingly, correlations at longer-range (i.e. outside the window of the defined spatial neighborhood) are also dramatically reduced.
Accordingly, the eigenspectrum of the locally whitened outputs is significantly flatter compared to the inputs (\autoref{fig:conv2d}D left vs.\ right columns).
We also provide an example using 1D inputs in \autoref{appendix:conv}.
This empirical result is not obvious --- that whitening individual \textit{overlapping} \textit{local} neighborhoods of neurons should produce a more \textit{globally} whitened output covariance.
Indeed, exactly how or when a globally whitened solution is possible from whitening of spatial overlapping neighborhoods of the inputs is a problem worth pursuing.

\vspace{-1.1em}

\section{Related Work}


\subsection{Biologically plausible whitening networks}

Biological circuits operate in the online setting and, due to physical constraints, must learn exclusively using local signals.
Therefore, to plausibly model neural computation, a neural network model must operate in the online setting (i.e., streaming data) and use local learning rules \citep{pehlevan_neuroscience-inspired_2019}.
There are a few existing normative models of adaptive statistical whitening and related transformations; however, these models use synaptic plasticity mechanisms (i.e., changing $\rmW$) to adapt to changing input statistics \citep{pehlevan2015normative,westrick2016pattern,chapochnikov2021normative,mlynarskiefficient2021, lipshutz2022normative}.
Adaptation of neural population responses to changes in sensory input statistics occurs rapidly, on the order of hundreds of milliseconds to seconds \citep{muller1999rapid,wanner2020whitening}, so it could potentially arise from short-term synaptic plasticity, which operates on the timescale of tens of milliseconds to minutes \citep{zucker2002short}, but not by long-term synaptic plasticity, which operates on the timescale of minutes or longer \citep{martin2000synaptic}. 
Here, we have proposed an alternative hypothesis: that modulation of neural gains, which operates on the order of tens of milliseconds to minutes \citep{ferguson2020mechanisms}, facilitates rapid adaptation of neural populations to changing input statistics.

%

\subsection{Tomography and ``sliced'' density measurements}

Leveraging 1D projections to compute the symmetric whitening transform is reminiscent of approaches taken in the field of tomography.
Geometrically, our method represents an ellipsoid (i.e., the $N$ dimensional covariance matrix) using noisy 1D projections of the ellipsoid onto axes spanned by frame vectors (i.e., estimates of the marginal variances).
This is a special case of reconstruction problems studied in geometric tomography \citep{karl_reconstructing_1994,gardner1995geometric}.
A distinction between tomography and our approach to symmetric whitening is that we are not reconstructing the multi-dimensional inputs; instead, we are utilizing the univariate measurements to transform an ellipsoid into a hyper-sphere.

In optimal transport, ``sliced'' methods offer a way to measure otherwise intractable $p$-Wasserstein distances in high dimensions \citep{bonneel2015sliced}, thereby enabling their use in optimization loss functions.
Sliced methods estimate Wasserstein distance by taking series of 1D projections of two densities, then computing the expectation over all 1D Wasserstein distances, for which there exists an analytic solution.
The 2-Wasserstein distance between a 1D zero-mean Gaussian with variance $\sigma^2$ and a standard normal density is
\begin{align*}
W_2\left(\mathcal{N}\left(0, \sigma^2\right) ; \mathcal{N}\left(0, 1\right)\right) = \left\|\sigma-1\right\|.
\end{align*}
This is strikingly similar to \autoref{eq:gupdate}.
However, distinguishing characteristics of our approach include:
1) minimizing distance between \textit{variances} rather than standard deviations;
2) directions along which we compute slices are fixed, whereas sliced methods compute a new set of projections at each optimization step;
3)  our network operates online, \textit{without} backpropagation.

\section{Discussion}

Our study introduces a recurrent circuit for adaptive whitening using \textit{gain modulation} to transform joint second-order statistics of their inputs based on \textit{marginal} variance measurements.
We demonstrate that, given sufficiently many marginal measurements along unique axes, the network produces symmetric whitened outputs.
Our objective (\autoref{eq:symmetricobjective}) provides a novel way to think about the classical problem of statistical whitening, and draws connections to old concepts from tomography and transport theory.
This framework is \textit{flexible and extensible}, with some possible generalizations explored in \autoref{appendix:applications}.
For example, we show that our model provides a way to prevent representational collapse in the analytically tractable example of online principal subspace learning (Appendix \ref{ssec:principal_subspace}).
Additionally, by replacing the unity marginal variance constraint by a set of target variances differing from 1, the network can be used to transform its input density to one matching the corresponding (non-white) covariance (Appendix \ref{ssec:cov_homeo}).

\subsection{Implications for machine learning}
Decorrelation and whitening are canonical transformations in signal processing, widely used in compression and channel coding. 
Deep nets are generally not trained to whiten, although their response variances are generally normalized during training through batch normalization, and recent methods \citep[e.g.][]{bardes2021vicreg} do impose global whitening properties in their objective functions.
Modulating feature gains has proven effective in adapting pre-trained neural networks to novel inputs with out-of-training distribution statistics \citep{balle_nonlinear_2020, duong_multi-rate_2022, mohan_adaptive_2021}.
Future architectures may benefit from adaptive run-time adjustments to changing input statistics \citep[e.g.][]{hu2021lora}. 
Our framework provides an unsupervised, online mechanism that avoids `catastrophic forgetting' in neural networks during continual learning.

\subsection{Implications for neuroscience}
It has been known for nearly 100 years \citep{adrian1928action} that single neurons rapidly adjust their sensitivity (gain) adaptively, based on recent response history.
Experiments suggest that neural populations \textit{jointly} adapt, adjusting both the amplitude of their responses, as well as their correlations \citep[e.g.][]{benucci2013adaptation, friedrich2013neuronal} to confer dynamic, efficient multi-channel coding.
The natural thought is that they achieve this by adjusting the strength of their interactions (synaptic weights).
Our work provides a \textit{fundamentally different} solution: these effects can arise solely through gain changes, thereby generalizing rapid and reversible single neuron adaptive gain modulation to the level of a neural population.

Support for our model will ultimately require careful experimental measurement and analysis of responses and gains of neurons in a circuit during adaptation \citep[e.g.][]{wanner2020whitening}. 
Our model predicts: 
1) Specific architectural constraints, such as reciprocally connected interneurons \citep{kepecs_fishell2014}, with consistency between their connectivity and population size (e.g.\ in the olfactory bulb). 
2) Synaptic strengths that remain stable during adaptation, which would adjudicate between our model and more conventional adaptation models relying on synaptic plasticity \citep[e.g.][]{lipshutz2022normative}. 
3) Interneurons that modulate their gains according to the difference between the variance of their post-synaptic inputs and some \textit{target variance} (\autoref{eq:gupdate}; also see Appendix \ref{ssec:cov_homeo}).
Experiments could assess whether interneuron input variances converge to the same values after adaptive whitening. 
4) Interneurons that increase their gains with the variance of their inputs (i.e. $\bar{z}_{i,t}^2$).
Input variance-dependent gain modulation may be mediated by changes in slow Na\textsuperscript{+} currents \citep{kim2003slow}.
This predicts a mechanistic role for interneurons during adaptation, and complements the observed gain effects found in excitatory neurons described in classical studies \citep{fairhallefficiency2001,nagel_temporal_2006}. 

\subsection{Conclusion}

Whitening is an effective constraint for preventing feature collapse in representation learning \citep{zbontar2021barlow, ermolov2021whitening}.
The networks developed here provide a whitening solution that is particularly well-suited for applications prioritizing streaming data and low-power consumption.

\section*{Acknowledgements}
We thank Pierre-Étienne Fiquet, Sina Tootoonian, Michael Landy, and the ICML reviewers for their helpful feedback.

\clearpage

\bibliography{references}
\bibliographystyle{icml2023}

\newpage
\appendix
\onecolumn

\section{Notation}
\label{appendix:notation}

For $N\ge2$, let $K_N:=N(N+1)/2$. Let $\R^N$ denote $N$-dimensional Euclidean space equipped with the Euclidean norm, denoted $\|\cdot\|_2$.
Let $\R_+^N$ denote the non-negative orthant in $\R^N$. 
Given $K\ge 2$, let $\R^{N\times K}$ denote the set of $N\times K$ real-valued matrices. Let $\sS^N$ denote the set of $N\times N$ symmetric matrices and let $\sS_{++}^N$ denote the set of $N\times N$ symmetric positive definite matrices.

Matrices are denoted using bold uppercase letters (e.g., $\rmM$) and vectors are denoted using bold lowercase letters (e.g., $\rvv$).
Given a matrix $\rmM$, $M_{ij}$ denotes the entry of $\rmM$ located at the $i$\textsuperscript{th} row and $j$\textsuperscript{th} column.
Let ${\bf 1}=[1,\dots,1]^\top$ denote the $N$-dimensional vector of ones.
Let $\rmI_N$ denote the $N\times N$ identity matrix.

Given vectors $\rvv,\rvw\in\R^N$, define their Hadamard product by $\rvv\circ\rvw:=(v_1w_1,\dots,v_Nw_N)\in\R^N$.
Define $\rvv^{\circ 2}:=(v_1^2,\dots,v_N^2)\in\R^N$.

Let $\langle\cdot\rangle_t$ denote expectation over $t=1,2,\dots$.

The $\diag{\cdot}$ operator, similar to \texttt{numpy.diag()} or MATLAB's \texttt{diag()}, can either: 1) map a vector in $\mathbb{R}^K$ to the diagonal of a $K \times K$ zeros matrix; or 2) map the diagonal entries of a $K \times K$ matrix to a vector in $\mathbb{R}^K$.
The specific operation being used should be clear by context.
For example, given a vector $\rvv\in\R^K$, define $\text{diag}(\rvv)$ to be the $K\times K$ diagonal matrix whose $(i,i)$\textsuperscript{th} entry is equal to $v_i$, for $i=1,\dots,K$.
Alternatively, given a sqaure matrix $\rmM\in\R^{K\times K}$, define $\text{diag}(\rmM)$ to be the $K$-dimensional vector whose $i$\textsuperscript{th} entry is equal to $M_{ii}$, for $i=1,\dots,K$.

\section{Optimal Solution to Symmetric Whitening Objective}\label{apdx:optimal}

In this section, we prove that the optimal solution to the optimization problem in \eqref{eq:symmetricobjectivevanilla} is given by $\rvy_t=\rmC_{xx}^{-1/2}\rvx_t$ for $t=1,\dots,T$ (we treat the case that $T<\infty$). 

We first recall Von Neumann's trace inequality \citep[see, e.g., ][Theorem 3.1]{carlsson2021neumann}.



\begin{lemma}[Von Neumann’s trace inequality]
Suppose $\rmA,\rmB\in\R^{n\times m}$ with $n\le m$. Let $\sigma_1^A\ge\cdots\ge\sigma_n^A\ge0$ and $\sigma_1^B\ge\cdots\ge\sigma_n^B\ge0$ denote the respective singular values of $\rmA$ and $\rmB$. Then
\begin{align*}
    \Tr(\rmA\rmB^\top)\le\sum_{i=1}^n\sigma_i^A\sigma_i^B.
\end{align*}
Furthermore, equality holds if and only if $\rmA$ and $\rmB$ share left and right singular vectors.
\end{lemma}


We can now proceed with the proof of our result. We first concatenate the inputs and outputs into data matrices $\rmX=[\rvx_1,\dots,\rvx_T]\in\R^{N\times T}$ and $\rmY=[\rvy_1,\dots,\rvy_T]\in\R^{N\times T}$. We can write \eqref{eq:symmetricobjectivevanilla} as follows:
\begin{align*}
    \min_{\rmY}\|\rmX-\rmY\|_F^2\qquad\text{subject to}\qquad\rmY\rmY^\top=T\rmI_N.
\end{align*}
Expanding, substituting in with the constraint $\rmY\rmY^\top=T\rmI_N$ and dropping terms that do not depend on $\rmY$ results in the objective
\begin{align*}
    \max_{\rmY}\Tr(\rmX\rmY^\top)\qquad\text{subject to}\qquad\rmY\rmY^\top=T\rmI_N.
\end{align*}
By Von Neumann’s trace inequality, the trace is maximized when the singular vectors of $\rmY$ are aligned with the singular vectors of $\rmX$. In particular, if the SVD of $\rmX$ is given by $\rmX=\rmU_x\rmS_x\rmV_x^\top$, then the optimal $\rmY$ is given by $\rmY=\sqrt{T}\rmU_x\rmV_x^\top$, which is precisely $\rmC_{xx}^{-1/2}\rmX$, where $\rmC_{xx}:=\tfrac1T\rmX\rmX^\top=\rmU_x\rmS_x^2\rmU_x^\top$.

\section{Proof of Proposition \ref{prop:marginal}}\label{appendix:num_projections}



\begin{proof}[Proof of Proposition \ref{prop:marginal}]
Suppose \autoref{eq:Cyy} holds. Then, for $i=1,\dots,K$,
\begin{align*}
    \langle(\rvw_i^\top\rvy_t)^2\rangle_t=\langle\rvw_i^\top\rvy_t\rvy_t^\top\rvw_i\rangle_t=\rvw_i^\top\rvw_i=1.
\end{align*}
Therefore, \autoref{eq:equivalence} holds. 

Now suppose \autoref{eq:equivalence} holds. Let $\rvv\in\R^N$ be an arbitrary unit vector. Then $\rvv\rvv^\top\in\sS^N$ and by \autoref{eq:spanSN}, there exist $g_1,\dots,g_K\in\R$ such that
\begin{align}\label{eq:vwspan}
    \rvv\rvv^\top=g_1\rvw_1\rvw_1^\top+\cdots+g_K\rvw_K\rvw_K^\top.
\end{align}
We have
\begin{align}\label{eq:v}
    \rvv^\top\langle\rvy_t\rvy_t^\top\rangle_t\rvv&=\Tr(\rvv\rvv^\top\langle\rvy_t\rvy_t^\top\rangle_t)=\sum_{i=1}^Kg_i\Tr(\rvw_i\rvw_i^\top\langle\rvy_t\rvy_t^\top\rangle_t)=\sum_{i=1}^Kg_i\Tr(\rvw_i\rvw_i^\top)=\Tr(\rvv\rvv^\top)=1.
\end{align}
The first equality is a property of the trace operator. The second and fourth equalities follow from \autoref{eq:vwspan} and the linearity of the trace operator. The third equality follows from \autoref{eq:equivalence}, the cyclic property of the trace, and the fact that each $\rvw_i$ is a unit vector. The final equality holds because $\rvv$ is a unit vector. Since \autoref{eq:v} holds for every unit vector $\rvv\in\R^N$, \autoref{eq:Cyy} holds.
\end{proof}



\section{Frame Factorizations of Symmetric Matrices}
\label{appendix:psd_frame}

\subsection{Analytic solution for the optimal gains}

Recall that the optimal solution of the symmetric objective in \autoref{eq:symmetricobjective} is given by $\rvy_t=\rmC_{xx}^{-1/2}\rvx_t$ for $t=1,2,\dots$. 
In our neural circuit with interneurons and gain control, the outputs of the primary neurons at equilibrium is (given in \autoref{eq:y_steadystate}, but repeated here for clarity),
\begin{align*}
\bar{\rvy}_t = \left[\rmI_N+\rmW\diag{\rvg}\rmW^\top \right]^{-1}\rvx_t,
\end{align*}
where $\rmW \in \R^{N\times K}$ is overcomplete, arbitrary (provided \autoref{eq:spanSN} holds), and \textit{fixed}; and elements of $\rvg\in\R^K$ can be interpreted as learnable scalar gains.
The circuit performs symmetric whitening when the gains $\rvg$ satisfy the relation
\begin{align}\label{eq:g_covx}
    \rmI_N+\rmW\diag{\rvg}\rmW^\top=\rmC_{xx}^{1/2}.
\end{align}
It is informative to contrast this with conventional approaches to symmetric whitening, which rely on eigendecompositions,
\begin{align*}
    \rmV\diag{{\boldsymbol \lambda}}^{1/2}\rmV^\top=\rmC_{xx}^{1/2},
\end{align*}
where $\rmV \in \R^{N \times N}$ and ${\boldsymbol \lambda}$ are the eigenvectors and eigenvalues of $\rmC_{xx}$, respectively.
Note that in this eigenvector formulation, both vector quantities (columns of $\rmV$) and scalar quantities (elements of ${\boldsymbol \lambda}$) need to be learned, whereas in our formulation (\autoref{eq:g_covx}), \textit{only scalars} need to be learned (elements of $\rvg$).

When $K \geq N(N+1)/2$, we can explicitly solve for the optimal gains $\rvg^\ast$ (derived in the next subsection):
\begin{align}\label{eq:g_opt}
    \rvg^\ast=\left[\left(\rmW^\top\rmW\right)^{\circ 2}\right]^{\dagger}\left[\rvw_1^\top\rmC_{xx}^{1/2}\rvw_1-1,\dots,\rvw_K^\top\rmC_{xx}^{1/2}\rvw_K-1\right]^\top.
\end{align}

\subsection{Isolating $\rvg$ embedded in a diagonal matrix}

In the upcoming subsection, our variable of interest, $\rvg$, is embedded along the diagonal of a matrix, then wedged between two fixed matrices, i.e. $\rmA_1 \diag{\rvg} \rmA_2$.
We employ the following identity to isolate $\rvg$,
\begin{align}\label{eq:diag_trick}
\diag{\rmA_1 \diag{\rvg} \rmA_2} = \left(\rmA_1 \circ \rmA_2^\top\right) \rvg,
\end{align}
where, on the left-hand-side, the inner $\diag{\cdot}$ forms a diagonal matrix from a vector, the outer $\diag{\cdot}$ returns the diagonal of a matrix as a vector, and $\circ$ is the element-wise Hadamard product.

\subsection{Deriving optimal gains}

Let $\rmC\in\mathbb{S}^N$, where $\mathbb{S}^N$ is the set of symmetric $N\times N$ matrices. Suppose $\rvg\in\R^K$ is such that the following holds:
\begin{align}
     \rmW\diag{\rvg} \rmW^\top &= \rmC\label{eq:decomp}
\end{align}
where $\rmW \in \mathbb{R}^{N \times K}$ is some fixed, arbitrary, frame with $K \geq\frac{N(N+1)}{2}$ (i.e. a representation that is $\mathcal{O}(N^2)$ overcomplete).
To solve for $\rvg$, we multiply both sides of \autoref{eq:decomp} from the left and right by $\rmW^\top$ and $\rmW$, respectively, then take the diagonal\footnote{Similar to commonly-used matrix libraries, the $\diag{\cdot}$ operator here is overloaded and can map a vector to a matrix or vice versa. See \autoref{appendix:notation} for details.} of the resultant matrices,
\begin{align}
     \diag{\rmW^\top \rmW \diag{\rvg} \rmW^\top \rmW }&= \diag{\rmW^\top \rmC \rmW}. \label{eq:pre_simplification}
\end{align}
Finally, employing the identity in \autoref{eq:diag_trick} yields 
\begin{align}
    (\rmW^\top \rmW)^{\circ 2} \rvg  &= \diag{\rmW^\top \rmC \rmW}, \label{eq:linear_system}\\
    \rvg &= \left[(\rmW^\top \rmW)^{\circ 2}\right]^{\dagger} \diag{\rmW^\top \rmC \rmW},
\end{align}
where $(\cdot)^{\circ 2}$ denotes element-wise squaring, $(\rmW^\top \rmW)^{\circ 2}$ is positive semidefinite by the Schur product theorem and $(\cdot)^\dagger$ denotes the Moore-Penrose pseudoinverse.
Thus, \textit{any} $N \times N$ symmetric matrix, 
can be encoded as a vector, $\rvg$, with respect to an arbitrary fixed frame, $\rmW$, by solving a standard linear system of $K$ equations of the form $\rmA\rvg = \rvb$.
Importantly, when $K=\frac{N(N+1)}{2}$ and the columns of $\rmW$ are not collinear, we have empirically found the matrix on the LHS, $(\rmW^\top\rmW)^{\circ 2}$, to be positive definite, so the vector $\rvg$ is uniquely defined.

Without loss of generality, assume that the columns of $\rmW$ are unit-norm (otherwise, we can always normalize them by absorbing their lengths into the elements of $\rvg$).
Furthermore, assume without loss of generality that $\rmC \in \mathbb{S}_{++}^N$, the set of all symmetric positive definite matrices (e.g. covariance, precision, PSD square roots, etc.).
When $\rmC$ is a covariance matrix, then $\diag{\rmW^\top \rmC \rmW}$ can be interpreted as a vector of projected variances of $\rmC$ along each axis spanned by $\rmW$.
Therefore, \autoref{eq:linear_system} states that the vector $\rvg$ is linearly related to the vector of projected variances via the element-wise squared frame Gramian, $(\rmW^\top \rmW)^{\circ 2}$.

\section{Saddle Point Property}\label{apdx:saddle}

In this section, we prove the following minimax property (for the case $t=1,\dots,T$ with $T$ finite):
\begin{align}\label{eq:interchange}
    \min_{\{\rvy_t\}}\max_\rvg\langle\ell(\rvx_t,\rvy_t,\rvg)\rangle_t = \max_\rvg\min_{\{\rvy_t\}}\langle\ell(\rvx_t,\rvy_t,\rvg)\rangle_t.
\end{align}
The proof relies on the following minimax property for a function that satisfies the saddle point property \citep[section 5.4]{boyd2004convex}.

\begin{theorem}\label{thm:saddle}
Let $V\subseteq\R^n$, $W\subseteq\R^m$ and $f:V\times W\to\R$. Suppose $f$ satisfies the saddle point property; that is, there exists $(\rva^\ast,\rvb^\ast)\in V\times W$ such that
\begin{align*}
    f(\rva^\ast,\rvb)\le f(\rva^\ast,\rvb^\ast)\le f(\rva,\rvb^\ast),\qquad\text{for all }(\rva,\rvb)\in V\times W.
\end{align*}
Then
\begin{align*}
    \min_{\rva\in V}\max_{\rvb\in W} f(\rva,\rvb)=\max_{\rvb\in W}\min_{\rva\in V} f(\rva,\rvb)=f(\rva^\ast,\rvb^\ast).
\end{align*}
\end{theorem}

In view of Theorem \ref{thm:saddle}, it suffices to show there exists $(\rvy_1^\ast,\dots,\rvy_T^\ast,\rvg^\ast)$ such that
\begin{align}\label{eq:saddle}
    \ell(\rvy_1^\ast,\dots,\rvy_T^\ast,\rvg) \le\ell(\rvy_1^\ast,\dots,\rvy_T^\ast,\rvg^\ast) \le\ell(\rvy_1,\dots,\rvy_T,\rvg^\ast),\qquad\text{for all }\rvy_1,\dots,\rvy_T\in\R^N\text{ and }\rvg\in\R^K.
\end{align}
Define $\rvy_t^\ast:=\rmC_{xx}^{-1/2}\rvx_t$ for all $t=1,\dots,T$ and define $\rvg^\ast$ as in \eqref{eq:g_opt} so that \eqref{eq:g_covx} holds.
Then, for all $\rvg\in\R^K$,
\begin{align*}
    \ell(\rvy_1^\ast,\dots,\rvy_T^\ast,\rvg)= \frac1T\sum_{t=1}^T\|\rvx_t-\rvy_t^\ast\|_2^2.
\end{align*}
Therefore, the first inequality in \eqref{eq:saddle} holds (in fact it is an equality for all $\rvg$). Next, we have
\begin{align*}
    \ell(\rvy_1,\dots,\rvy_T,\rvg^\ast)&=\frac1T\sum_{t=1}^T\|\rvx_t-\rvy_t\|_2^2+\frac1T\sum_{t=1}^T\Tr\left[\rmW\text{diag}(\rvg^\ast)\rmW^\top(\rvy_t\rvy_t^\top-\rmI_N)\right]\\
    &=\frac1T\sum_{t=1}^T(\rvx_t^\top\rvx_t-2\rvx_t^\top\rvy_t)+\frac1T\sum_{t=1}^T\Tr\left[(\rmI_N+\rmW\text{diag}(\rvg^\ast)\rmW^\top)(\rvy_t\rvy_t^\top-\rmI_N)\right]\\
    &=\frac1T\sum_{t=1}^T(\rvx_t^\top\rvx_t-2\rvx_t^\top\rvy_t)+\frac1T\sum_{t=1}^T\Tr\left[\rmC_{xx}^{1/2}(\rvy_t\rvy_t^\top-\rmI_N)\right]
\end{align*}
Since $\rmC_{xx}^{1/2}$ is positive definite, $\ell(\rvy_1,\dots,\rvy_T,\rvg^\ast)$ is strictly convex in $(\rvy_1,\dots,\rvy_T)$ with its unique minimum obtained at $\rvy_t=\rmC_{xx}^{-1/2}\rvx_t$ for all $t=1,\dots,T$ (to see this, differentiate with respect to $\rvy_1,\dots,\rvy_T$, set the derivatives equal to zero and solve for $\rvy_1,\dots,\rvy_T$). This establishes the second inequality in \eqref{eq:saddle} holds. Therefore, by Theorem \ref{thm:saddle}, \eqref{eq:interchange} holds.

\section{Weighted Average Update Rule for $\rvg$}
\label{appendix:generalizations}
The update for $\rvg$ in \autoref{eq:gupdate} can be generalized to allow for a weighted average over past samples. In particular, the general update is given by
\begin{align*}
    \rvg\gets\rvg+\eta\left(\frac{1}{Z}\sum_{s=1}^t\gamma^{t-s}\rvz_s^{\circ 2}-{\bf 1}\right),
\end{align*}
where $\gamma\in[0,1]$ determines the decay rate and $Z:=1+\gamma+\cdots+\gamma^{t-1}$ is a normalizing factor.

\section{Batched and Offline Algorithms for Whitening with RNNs via Gain Modulation} \label{appendix:algorithms}

In addition to the fully-online algorithm provided in the main text (Algorithm~\ref{alg:online}), we also provide two variants below.
In many applications, streaming inputs arrive in batches rather than one at a time (e.g. video streaming frames).
Similarly for conventional offline stochastic gradient descent training, data is sampled in batches.
Algorithm~\ref{alg:batched} would be one way to accomplish this in our framework, where the main difference between the fully online version is taking the mean across samples in the batch to yield average gain update $\Delta\rvg$ term.
Furthermore, in the fully offline setting when the covariance of the inputs, $\rmC_{xx}$ is known, Algorithm~\ref{alg:offline} presents a way to whiten the covariance directly.

\begin{minipage}{0.46\textwidth}
\begin{algorithm}[H]
\caption{Batched symmetric whitening}
\label{alg:batched}
\begin{algorithmic}[1]
\STATE {\bfseries Input:} Data matrix $\rmX \in \R^{N \times T}$ (centered)
\STATE {\bfseries Initialize:} $\rmW\in\R^{N\times K}$; $\rvg\in\R^K$; $\eta$; batch size $B$\\
\WHILE{not converged} 
\STATE $\rmX_B \gets \texttt{sample\_batch}(\rmX, B)$ 
\STATE $\rmY_B \gets [\rmI_N +\rmW \diag{\rvg} \rmW^\top]^{-1}\rmX_B$ 
\STATE $\rmZ_B \gets \rmW^\top\rmY_B$
\STATE $\Delta \rvg \gets \frac1T\text{diag}(\rmZ_B\rmZ_B^\top) - {\bf 1}$ 
\STATE $\rvg \gets \rvg + \eta \texttt{ mean}(\Delta \rvg, \texttt{axis=}1)$ 
 \ENDWHILE
\end{algorithmic}
\end{algorithm}
\end{minipage}
\hfill
\begin{minipage}{0.46\textwidth}
\begin{algorithm}[H]
\caption{Offline symmetric whitening}
\label{alg:offline}
\begin{algorithmic}[1]
\STATE {\bfseries Input:} Input covariance $\rmC_{xx}$
\STATE {\bfseries Initialize:} $\rmW\in\R^{N\times K}$; $\rvg\in\R^K$; $\eta$ \\
\WHILE{not converged} 
\STATE $\rmM \gets [\rmI_N + \rmW \diag{\rvg} \rmW^\top]^{-1}$
\STATE $\rmC_{yy} \gets \rmM \rmC_{xx} \rmM$
\STATE $\Delta \rvg \gets\diag{\rmW^\top \rmC_{yy} \rmW}-{\bf 1}$
\STATE $\rvg \gets \rvg + \eta \Delta \rvg$
 \ENDWHILE
\end{algorithmic}
\end{algorithm}
\end{minipage}

\section{Normalizing Ill-conditioned Inputs with Non-negative Constrained Gains}\label{appendix:nonneg}

\subsection{Quantifying whitening error}

Whitening with non-negative gains does not, in general, produce an output with identity covariance matrix; therefore, quantifying algorithm performance with the error defined in the main text would not be informative.
Because this extension shares similarities with ideas of regularized whitening, in which principal axes whose eigenvalues are below a certain threshold are unaffected by the whitening transform, we quantify algorithmic performance using thresholded Spectral Error,
\begin{align*}
   \text{Spectral Error}:= \frac 1 N \sum_i^N \max(\lambda_i-1, 0)^2,
\end{align*}
where $\lambda_i$ is the $i$\textsuperscript{th} eigenvalue of $\rmC_{yy}$.
Here, as in the main text, we set the threshold to 1.
\autoref{fig:loss_low_rank} shows that this network reduces spectral error.
Importantly, the converged solution depends on the initial choice of frame (see next subsection).

\begin{figure}[htb]
    \centering
    \includegraphics[width=.75\textwidth]{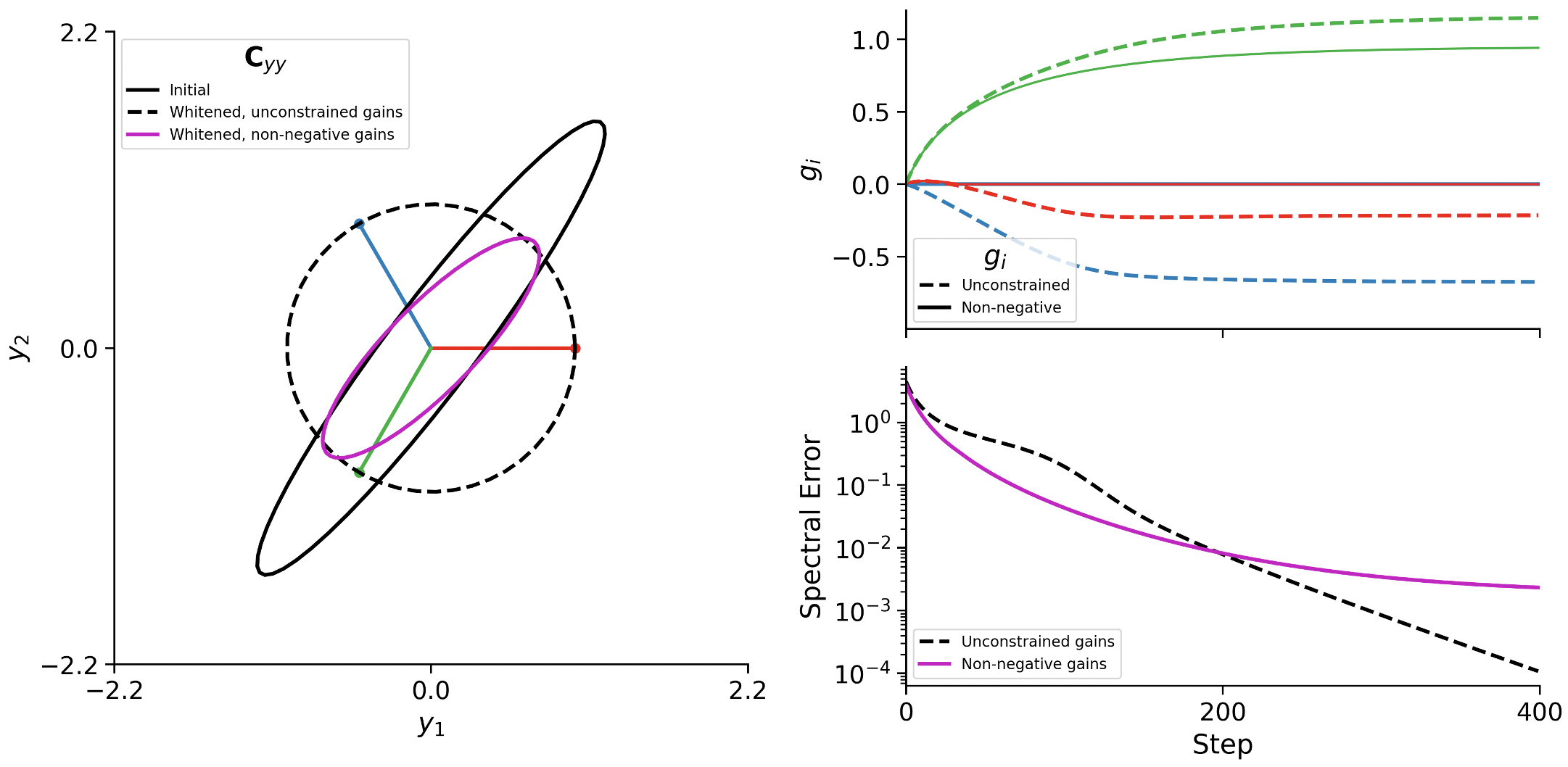}
    \caption{Whitening ill-conditioned inputs with non-negative gains. 
    {\bf A)} An equi-angular frame (red, blue, green; see Sec.~\ref{ssec:convergence}) whitening ill-conditioned inputs. 
    {\bf B)} Gains as algorithm progresses, using updates with either rectified or unrectified constraints.
    {\bf C)} Spectral Error (see text).}
    \label{fig:loss_low_rank}
\end{figure}

\subsection{Geometric intuition behind thresholded whitening with non-negative gains}
In general, the modified objective with rectified gains (\autoref{eq:grectified}) does not statistically whiten the inputs $\rvx_1,\rvx_2,\dots$, but rather adapts the non-negative gains $g_1,\dots,g_K$ to ensure that the variances of the outputs $\rvy_1,\rvy_2,\dots$ in the directions spanned by the frame vectors $\{\rvw_1,\dots,\rvw_K\}$ are bounded above by unity (Figure \ref{fig:nonneg}). 
This one-sided normalization carries interesting implications for how and when the circuit statistically whitens its outputs, which can be compared with experimental observations. 
For instance, the circuit performs symmetric whitening if and only if there are non-negative gains such that \autoref{eq:g_covx} holds (see, e.g., the top right example in Figure \ref{fig:nonneg}), which corresponds to cases such that the matrix $\rmC_{xx}^{1/2}$ is an element of the following cone (with its vertex translated by $\rmI_N$):
\begin{align*}
    \left\{\rmI_N+\sum_{i=1}^Kg_i\rvw_i\rvw_i^\top:\rvg\in\R_+^K\right\}.
\end{align*}
On the other hand, if the variance of an input projection is less than unity --- i.e., $\rvw_i^\top\rmC_{xx}\rvw_i\le1$ for some $i$ --- then the corresponding gain $g_i$ remains zero. When this is true for all $i=1,\dots,K$, the gains all remain zero and the circuit output is equal to its input (see, e.g., the bottom middle panel of Figure \ref{fig:nonneg}).

\begin{figure}[htb]
    \centering
    \includegraphics[width=.6\textwidth]{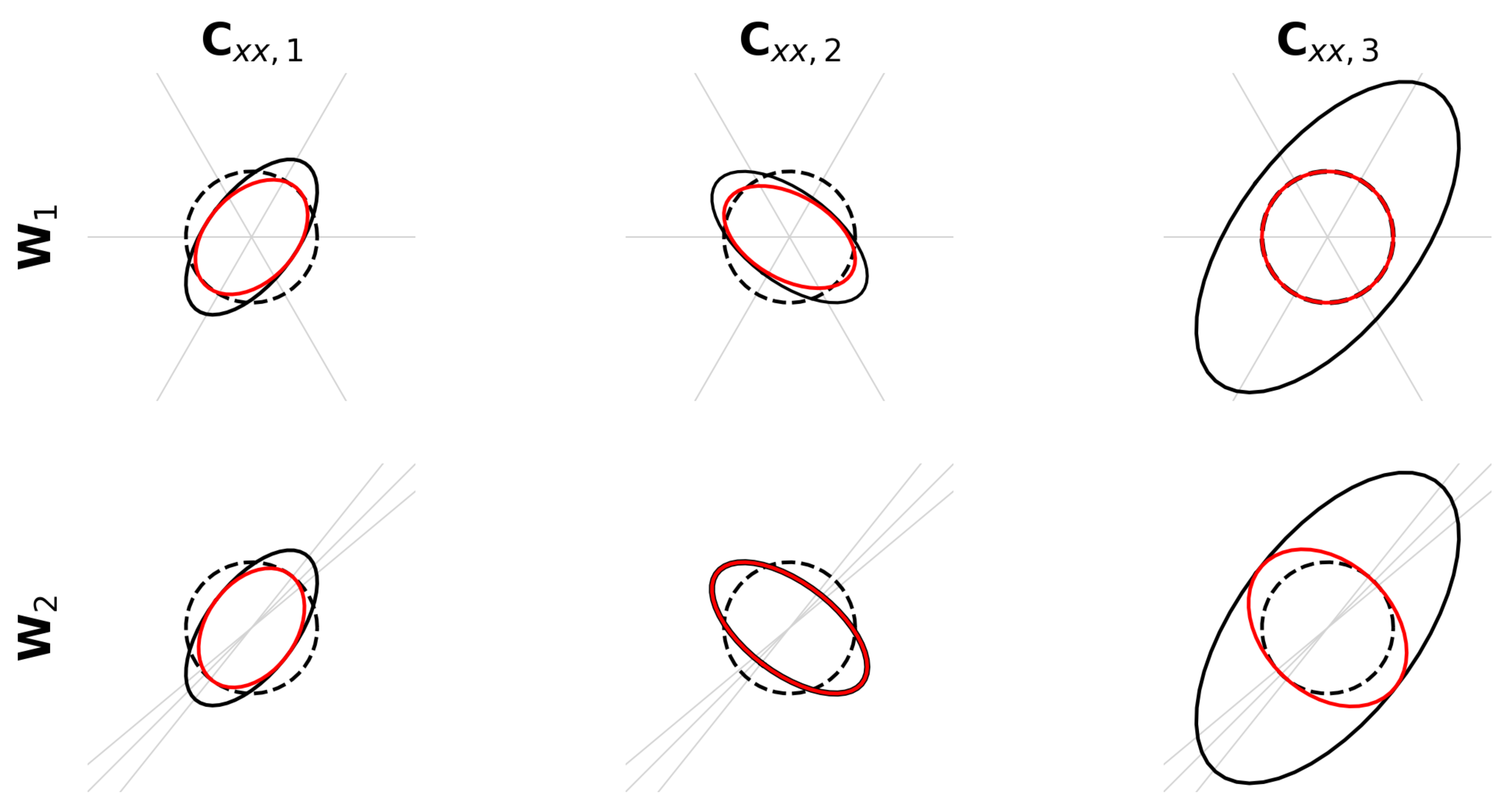}
    \caption{Geometric intuition of whitening with/without inequality constraint.
    Whitening efficacy using non-negative gains depends on $\rmW$ and $\rmC_{xx}$.
For $N=2$ and $K=3$, examples of covariance matrices $\rmC_{yy}$ (red ellipses) corresponding to optimal solutions $\rvy$ of objective \ref{eq:symmetricupper}, for varying input covariance matrices $\rmC_{xx}$ (black ellipses) and frames $\rmW$ (spanning axes denoted by gray lines).
Unit circles, which correspond to the identity matrix target covariance, are shown with dashed lines. 
Each row corresponds to a different frame $\rmW$ and each column corresponds to a different input covariance $\rmC_{xx}$.
    }
    \label{fig:nonneg}
\end{figure}

\section{Whitening Spatially Local Neighborhoods}\label{appendix:conv}

\subsection{Spatially local whitening in 1D}
For an $N$-dimensional input, we consider a network that whitens spatially local neighborhoods of size $M<N$.
To this end, we can construct $N$ filters of the form
\begin{align*}
    \rvw_i=\rve_i,\qquad i=1,\dots,N
\end{align*}
and $M(N-\frac{M+1}{2})$ filters of the form
\begin{align*}
    \rvw_{ij}=\frac{\rve_i+\rve_j}{\sqrt{2}},\qquad i,j=1,\dots,N,\qquad 1\le |i-j|\le M.
\end{align*}
The total number of filters is $(M+1)(N-\frac M2)$, so for fixed $M$ the number of filters scales linearly in $N$ rather than quadratically.

We simulated a network comprising $N=10$ primary neurons, and a convolutional weight matrix connecting each interneuron to spatial neighborhoods of three primary neurons.
Given input data with covariance $\rmC_{xx}$ illustrated in \autoref{fig:conv1d}A (left panel), this modified network succeeded to statistically whiten local neighborhoods of size of primary 3 neurons (right panel).  
Notably, the eigenspectrum (\autoref{fig:conv1d}B) after local whitening is much closer to being equalized.
Furthermore, while the global whitening solution produced a flat spectrum as expected, the local whitening network did not amplify the axis with very low-magnitude eigenvalues (\autoref{fig:conv1d}B right panel).

\begin{figure}[htb]
    \centering
    \includegraphics[width=.65\textwidth]{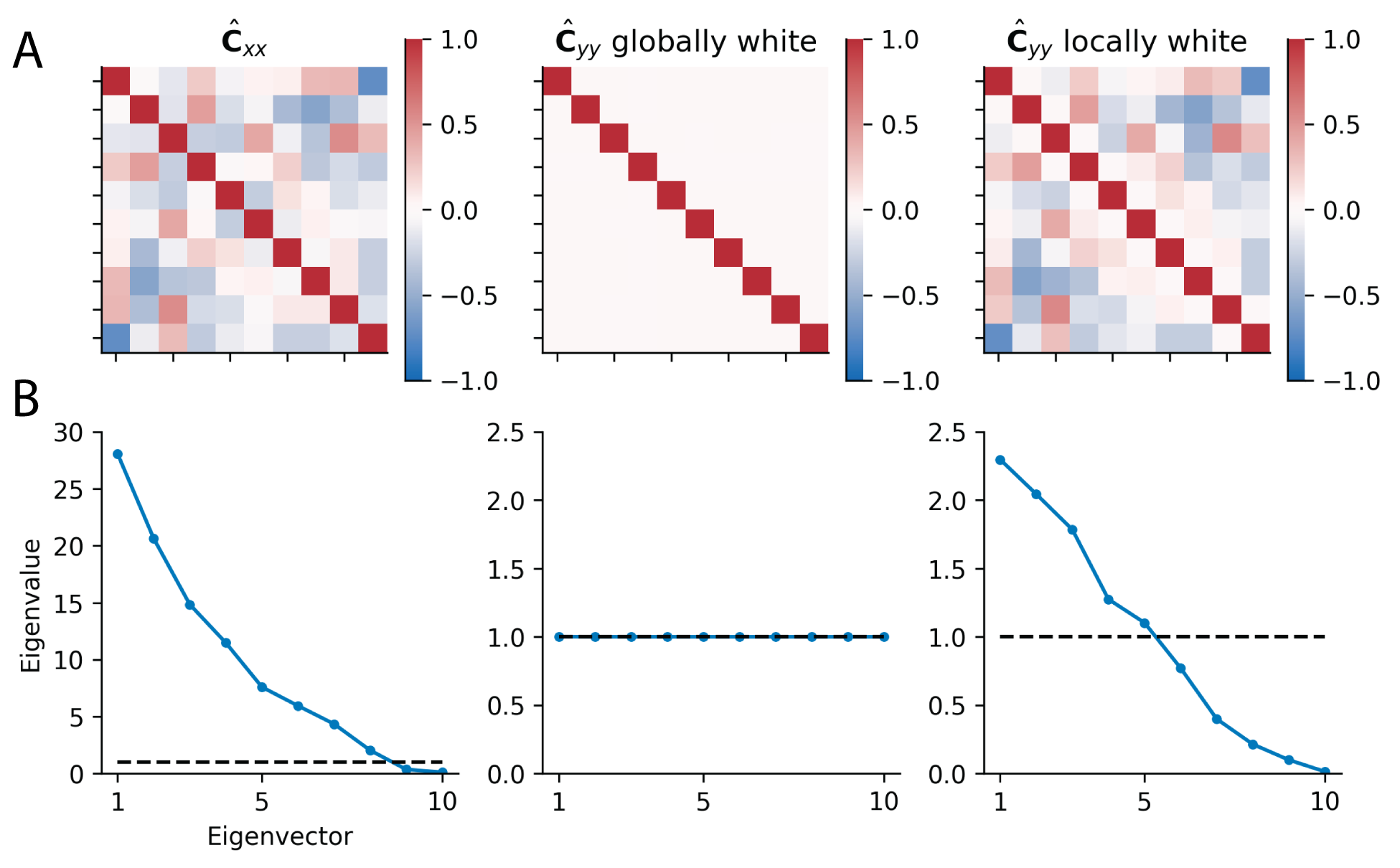}
\caption{ Statistically adapting local neighborhoods of neurons.
{\bf A)} 
$\hat\rmC_{xx}$ denotes correlation matrix, which are shown here for display purposes only, to facilitate comparisons.
Network with 10-dimensional input correlation (left) 10-dimensional output correlation matrix after global whitening (middle); and output correlation matrix after statistically whitening local neighborhoods of size 3.
The output correlation matrix of the locally adapted circuit has block-identity structure along the diagonal.
{\bf B)}  Corresponding eigenspectra of \textit{covariance} matrices of unwhitened (left), global whitened (middle), and locally whitened (right) network outputs.
The y-axis limits of the middle and right columns are the same, but different than the left column.
The black dashed line denotes unity.
} 
\label{fig:conv1d}
\end{figure}

\subsection{Filter bank  construction in 2D}
Here, we describe one way of constructing a set of convolutional weights for overlapping spatial neighborhoods (e.g. image patches) of neurons.
Given an $n \times m$ input and overlapping neighborhoods of size $h\ \times w$ to be statistically whitened, the samples are therefore matrices $X\in\R^{n \times m}$. 
In this case, filters $\rvw \in \R^{1 \times n \times m}$ can be indexed by pairs of pixels that are in the same patch:
\begin{align*}
    ((i,j),(k,\ell)),&&1\le i\le n,\qquad1\le j\le m,\qquad0\le|i-k|\le h,\qquad 0\le|j-\ell|\le w
\end{align*}
We can then construct the filters as,
\begin{align*}
    \rvw_{(i,j),(k,\ell)}(X)=
    \begin{cases}
        x_{i,j}&\text{if }(i,j)=(k,\ell),\\
        \frac{x_{i,j}+x_{k,\ell}}{\sqrt{2}}&\text{if }(i,j)\ne (k,\ell).
    \end{cases}
\end{align*}
In this case there are 
\begin{align*}
nm + wh\left[(n-w)(m-h) + (n-w)\frac{(h+1)}{2} + (m-h)\frac{(w+1)}{2} + (h+1)\frac{(w+1)}{2}\right]
\end{align*}
such filters, so the number of filters required scales linearly with $nm$ rather than quadratically.

\section{Additional Applications}\label{appendix:applications}

\subsection{Preventing representational collapse in online principal subspace learning}\label{ssec:principal_subspace}

Here, similar to \citet{lipshutz2022normative}, we show how whitening can prevent representational collapse using the analytically tractable example of online principal subspace learning.
Recent approaches to self-supervised learning have used decorrelation transforms such as whitening to prevent collapse during training \citep[e.g.][]{zbontar2021barlow}.
Future architectures may benefit from online, adaptive whitening to allow for continual learning and test-time adaptation.

Consider a primary neuron whose \textit{pre-synaptic} input at time $t$ is $\rvs_t \in \mathbb{R}^D$, and corresponding output is $y_t:=\rvv^{\top} \rvs_t$, where $\rvv \in \mathbb{R}^D$ are the synaptic weights connecting the inputs to the neuron. 
An online variant of power iteration algorithm learns the top principal component of the inputs by updating the vector $\mathbf{v}$ as follows:
\begin{align*}
\rvv &\leftarrow \rv+\zeta\left(y_t \rvs_t-y_t^2 \mathbf{v}\right)\\
\rvv &\leftarrow \frac{1}{\Vert \rvv \Vert}{\rvv}
\end{align*}
where $\zeta>0$ is small.

Next, consider a population of $2 \leq N \leq D$ primary neurons with outputs $\mathbf{y}_t \in \mathbb{R}^N$ and feedforward synaptic weight vectors $\mathbf{v}_1, \ldots, \mathbf{v}_N \in \mathbb{R}^D$ connecting the pre-synaptic inputs $\mathbf{s}_t$ to the $N$ neurons. 
Running $N$ parallel instances of the power iteration algorithm defined above \textit{without} a decorrelation process results in representational collapse, because each synaptic weight vector $\mathbf{v}_i$ converges to the top principal component (\autoref{fig:subspace}, orange). 
We demonstrate that our whitening algorithm via gain modulation readily solves this problem.
Here, it is important that the whitening happen on a faster timescale than the principal subspace learning, to avoid collapse \citep[see][for details]{lipshutz2022normative}.

For this simulation, we set $D=3, N=2$ and randomly sample i.i.d. pre-synaptic inputs $\mathbf{s}_t \sim \mathcal{N}(\mathbf{0}, \operatorname{diag}(5,2,1))$. 
We randomly initialize two vectors $\mathbf{v}_1, \mathbf{v}_2 \in \mathbb{R}^3$ with i.i.d. Gaussian entries. 
At each time step $t$, we project pre-synaptic inputs to form the post-synaptic primary neuron inputs, $\mathbf{x}_t:=\left[\mathbf{v}_1^{\top} \mathbf{s}_t, \mathbf{v}_2^{\top} \mathbf{s}_t\right]^\top$, forming the input to Algorithm~\ref{alg:online}.
Let $\mathbf{y}_t$ be the primary neuron steady-state output; that is, $\mathbf{y}_t=\left(\rmI_N + \rmW \diag{\rvg} \rmW^{\top}\right)^{-1} \mathbf{x}_t$ (\autoref{eq:y_steadystate}). 
For $i=1,2$, we update $\mathbf{v}_i$ according to the above-defined update rules, with $\zeta=10^{-3}$.
We update the gains $\rvg$ according to Algorithm~\ref{alg:online} with $\eta=10\zeta$.
To measure the online subspace learning performance, we define
$$
\text { Subspace error }:=\left\|\rmV\left(\rmV^{\top} \rmV\right)^{-1} \rmV^{\top}-\diag{[1,1,0]} \right\|_{\text {Frob }}^2, \quad \rmV:=\left[\rvv_1, \rvv_2\right] \in \mathbb{R}^{3 \times 2}
$$
\autoref{fig:subspace} (blue) shows that our adaptive whitening algorithm with gain modulation successfully facilitates subspace learning and prevents representational collapse.

\begin{figure}[ht]
    \centering
    \includegraphics[width=.5\textwidth]{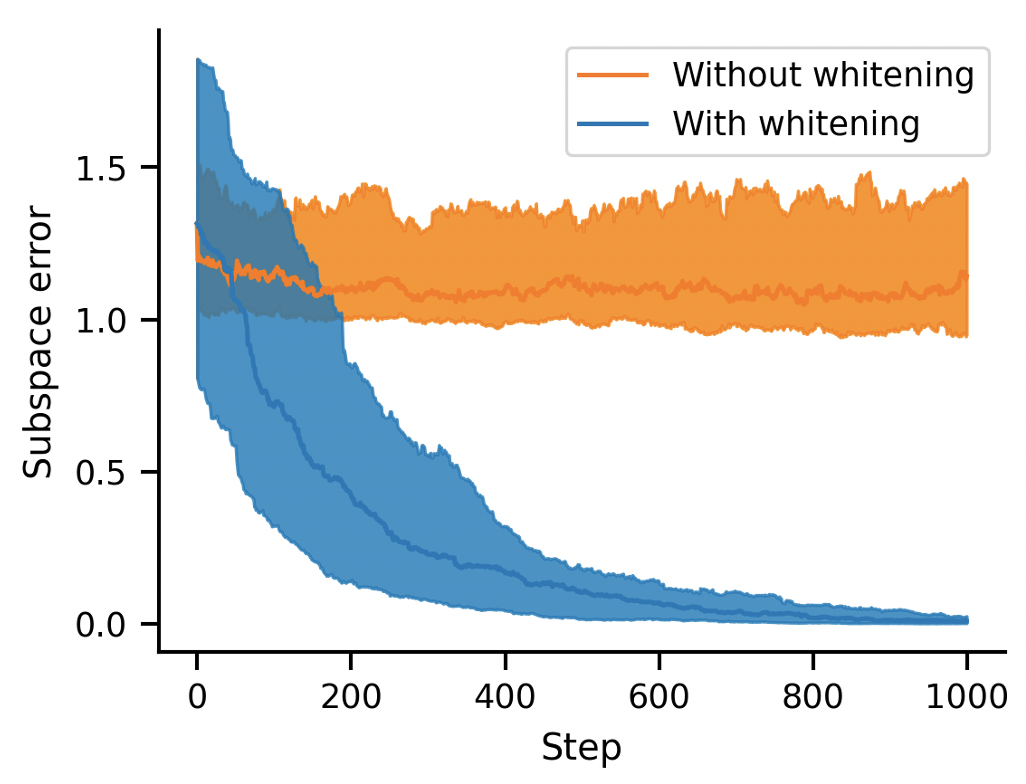}
\caption{Adaptive symmetric whitening with gain modulation prevents representational collapse during online principal subspace learning. Without whitening, subspace error stabilizes at a non-zero value, indicating that the network has converged to a collapsed representation. Shaded curves are median and [25\%, 75\%] quantiles over 50 random intializations.} 
\label{fig:subspace}
\end{figure}

\subsection{Generalized adaptive covariance transformations}\label{ssec:cov_homeo}

Our framework for adaptive whitening via gain modulation can easily be generalized to adaptively transform a signal with some initial covariance matrix to one with \textit{any target covariance} (i.e. not just the identity matrix).
This demonstrates that our adaptive gain modulation framework has implications beyond statistical whitening.
This could, for example, allow online systems to stably maintain some initial/target (non-white) output covariance under changing input statistics \citep[i.e. covariance homeostasis,][]{westrick2016pattern,benucci2013adaptation}.
The key insight, similar to the main text, is that a full-rank covariance matrix has $K_N$ degrees of freedom, and therefore marginal measurements along $K_N$ distinct axes is necessary and sufficient to represent the matrix \citep{karl_reconstructing_1994}.

Let $\rmC_\text{target}$ be some arbitrary target covariance matrix. Then the general objective is
\begin{align}\label{eq:genobjectivevanilla}
    &\min_{\{\rvy_t\}}\langle\|\rvx_t-\rvy_t\|_2^2\rangle_t\quad\text{s.t.}\quad\langle\rvy_t \rvy_t^\top\rangle_t=\rmC_\text{target}.
\end{align}
Following the same logic as in the main text, the Lagrangian becomes
\begin{align}\label{eq:genLagrange}
    &\max_\rvg\min_{\{\rvy_t\}}\langle\ell(\rvx_t,\rvy_t,\rvg)\rangle_t,\\
    &\text{where}\enspace\ell(\rvx,\rvy,\rvg):=\|\rvx-\rvy\|_2^2+\sum_{i=1}^Kg_i\left\{(\rvw_i^\top\rvy)^2-\sigma^2_i\right\},\nonumber
\end{align}
where $\sigma_i^2=\rvw_i^\top\rmC_\text{target}\rvw_i$ is the marginal variance along the axis spanned by $\rvw_i$. 
When $\rmC_\text{target}=\rmI_N$, then $\sigma_i^2=1$ for all $i$, and this reduces to our original overcomplete whitening objective (\autoref{eq:symmetricobjective}).
The only difference in the recursive algorithm optimizing this generalized objective is the gain update rule,
\begin{align}\label{eq:gengupdate}
    g_i& \gets g_i+\frac\eta2\nabla_{g_i} \ell(\rvx_t,\bar\rvy_t,\rvg) \nonumber \\
    &=g_i+\eta\left(\bar z_{i,t}^2-\sigma^2_i \right).
\end{align}
We can interpret this formulation as each interneuron having a pre-determined target input variance (perhaps learned over long time-scales), and adjusting its gains to modulate the joint responses of the primary neurons until its input variance matches the target.

\end{document}

%% file: math_commands.tex
\usepackage{amsmath,amsfonts,bm}
\usepackage{amsthm}

\newcommand{\diag}[1]{\operatorname{diag}\left( #1\right)}

\def\eqref#1{equation~\ref{#1}}

\def\1{\bm{1}}

\def\rv{{\textnormal{v}}}

\def\rva{{\mathbf{a}}}
\def\rvb{{\mathbf{b}}}

\def\rve{{\mathbf{e}}}

\def\rvg{{\mathbf{g}}}

\def\rvs{{\mathbf{s}}}

\def\rvv{{\mathbf{v}}}
\def\rvw{{\mathbf{w}}}
\def\rvx{{\mathbf{x}}}
\def\rvy{{\mathbf{y}}}
\def\rvz{{\mathbf{z}}}

\def\rmA{{\mathbf{A}}}
\def\rmB{{\mathbf{B}}}
\def\rmC{{\mathbf{C}}}

\def\rmI{{\mathbf{I}}}

\def\rmM{{\mathbf{M}}}

\def\rmS{{\mathbf{S}}}

\def\rmU{{\mathbf{U}}}
\def\rmV{{\mathbf{V}}}
\def\rmW{{\mathbf{W}}}
\def\rmX{{\mathbf{X}}}
\def\rmY{{\mathbf{Y}}}
\def\rmZ{{\mathbf{Z}}}

\DeclareMathAlphabet{\mathsfit}{\encodingdefault}{\sfdefault}{m}{sl}
\SetMathAlphabet{\mathsfit}{bold}{\encodingdefault}{\sfdefault}{bx}{n}

\def\sS{{\mathbb{S}}}

\newcommand{\R}{\mathbb{R}}

\DeclareMathOperator{\Tr}{Tr}